\documentclass[12pt]{article}

\usepackage{fullpage}
\usepackage{cancel}
\usepackage[T1]{fontenc}
\usepackage{ae, aecompl}
\usepackage{graphicx}
\usepackage{amssymb}
\usepackage{amsmath}
\usepackage{subfigure}
\usepackage{enumitem}
\usepackage[round,authoryear]{natbib}
\usepackage{color}
\usepackage{array}
\usepackage{rotating}
\usepackage{colortbl}
\usepackage{tikz}
\usepackage{titling}
\usetikzlibrary{patterns}

\definecolor{webgreen}{rgb}{0,0.4,0}
\definecolor{webbrown}{rgb}{0.6,0,0}
\definecolor{purple}{rgb}{0.5,0,0.25}
\definecolor{darkblue}{rgb}{0,0,0.7}
\definecolor{darkred}{rgb}{0.7,0,0}
\definecolor{darkgreen}{rgb}{0,0.7,0}
\usepackage[pdfborder=false]{hyperref}
\hypersetup{colorlinks,citecolor=darkred,filecolor=black,linkcolor=darkblue,urlcolor=webgreen,pdfpagemode=none, pdfstartview=FitH}
\newcommand{\ignore}[1]{}
\newtheorem{lemma}{{\sc Lemma}}

\newtheorem{cor}{{\sc Corollary}}
\newtheorem{theorem}{{\sc Theorem}}
\newtheorem{defn}{{\sc Definition}}

\newtheorem{step}{{\sc Step}}

\newtheorem{fact}{{\sc Fact}}

\definecolor{webgreen}{rgb}{0,0.4,0}
\definecolor{webbrown}{rgb}{0.6,0,0}
\definecolor{purple}{rgb}{0.5,0,0.25}
\definecolor{darkblue}{rgb}{0,0,0.7}
\definecolor{darkred}{rgb}{0.7,0,0}
\definecolor{darkgreen}{rgb}{0,0.7,0}

\newenvironment{proof}{\noindent {\bf \sl Proof\/}:\enspace}
{\hfill $\blacksquare{}$ \vspace{12pt}}
\newenvironment{proof*}{\noindent {\sl Proof\/}:\enspace}
{\hfill $\square{}$ \vspace{12pt}}
\usepackage{sectsty}

\def\qed{\hfill $\Box$}

\begin{document}
\title{{\bf Equity in auction design with
unit-demand agents and non-quasilinear preferences}~\thanks{We are grateful to Anna Bogomolnaia, Masahiro Kawasaki, Herv\'{e} Moulin, Shunya Noda, Daisuke Oyama, Yuya Wakabayashi, and Yu Zhou. We gratefully acknowledge financial support from the Joint
Usage/Research Center at ISER, Osaka University and the Japan Society for the Promotion of Science (Kazumura, 14J05972; Serizawa, 15J01287, 15H03328, 15H05728). Debasis Mishra acknowledges hospitality and support of ISER, Osaka University.}}

\author{Tomoya Kazumura,~Debasis Mishra, and~Shigehiro Serizawa~\thanks{Kazumura: Graduate School of Economics,
Kyoto University, \texttt{pge003kt@gmail.com}; Mishra: Indian Statistical Institute, Delhi, \texttt{dmishra@gmail.com};
Serizawa: Faculty of Economics, Osaka University of Economics, \texttt{serizawa.8558@gmail.com}}}

\maketitle

\begin{abstract}
We study a model of auction design where a seller is selling a set of objects to a set of agents who can be assigned
no more than one object. Each agent's preference over (object, payment) pair need not be quasilinear.
If the domain contains all classical preferences,
we show that there is a unique mechanism, the minimum Walrasian equilibrium price (MWEP) mechanism, which is strategy-proof, individually rational, and satisfies equal
treatment of equals, no-wastage (every object is allocated to some agent), and
no-subsidy (no agent is subsidized). This provides an equity-based characterization of the MEWP mechanism, and complements the efficiency-based characterization of the MWEP mechanism known in the literature.
\end{abstract}

\newpage

\section{Introduction}
\label{sec:intro}

Many auctions in practice are conducted by public bodies such as governments, local authorities, and public institutions.
Prominent examples include spectrum auctions and other public allocation problems involving exclusive rights, such as the allocation of ownership or development rights.
In these contexts, favoring particular firms or individuals is neither permissible nor desirable, and is likely to attract public criticism.\footnote{
For instance, in the context of spectrum auctions, fairness is recognized as an important criterion for choosing auction formats \citep{Mcmillan95,Kwerel01}.}
Thus, in addition to standard incentive and participation constraints, fairness can be viewed as an additional constraint in such design problems.

We consider an environment in which there are multiple heterogeneous objects with unit-demand agents.
Preferences of agents may not be quasilinear, and thus, may exhibit income effects.
There are various notions of fairness, and the appropriate fairness notion may depend on the application.
To cover as wide a range of applications as possible, we
employ one of the weakest notions of fairness - {\it equal treatment of equals} (ETE). 
This criterion requires that if two agents have identical preferences, then the welfare levels they obtain from the mechanism are identical.
This principle reflects Aristotle's conception
of justice; he writes:\footnote{%
Aristotle, \textquotedblleft\ Politics,\textquotedblright\ Book III: The
theory of citizenship and constitutions, C: The principle of oligarchy and
democracy and the nature of distributional justice, Chapter 9.}

\begin{quote}
Justice is considered to mean equality. It does not mean equality - but
equality for those who are equal, and not for all.
\end{quote}
ETE can also be justified as a necessary condition for standard
fairness desiderata such as no envy and anonymity. 

We impose {\it strategy-proofness} (dominant strategy incentive compatibility) and {\it individual rationality} as incentive constraints.
In addition, we consider two minor properties, no subsidy and no wastage.
{\it No subsidy} requires agents' payments to be nonnegative, which is a standard assumption in auction design.
{\it No wastage} requires that every object be assigned to some agent.\footnote{In our model, we assume that the number of agents is larger than that of objects.}

In our model, the minimum Walrasian equilibrium price (MWEP) mechanism satisfies
equal treatment of equals, strategy-proofness, individual rationality, no subsidy,
and no wastage \citep{Demange85}.\footnote{We focus on deterministic mechanisms. Deterministic auction mechanisms are simple and transparent, and are often preferred by sellers over auctions that involve randomization.} 
Equal treatment of equals is satisfied by a large
class of mechanisms, and there also exist many mechanisms that satisfy
strategy-proofness, individual rationality, no subsidy, and no wastage. 
Nevertheless, we show that if the domain of preferences is sufficiently rich (i.e., contains all
{\em classical} preferences), the MWEP mechanism is the unique mechanism that satisfies equal
treatment of equals, strategy-proofness, individual rationality, no subsidy, and
no wastage (Theorem~\ref{theo:main}).

An important implication of our uniqueness result concerns efficiency.
The MWEP mechanism is efficient. As a consequence, our main result implies that
any mechanism that satisfies equal treatment of equals, strategy-proofness,
individual rationality, no subsidy, and no wastage must be efficient
(Corollary~\ref{cor:eff}). Fairness and efficiency are often viewed as competing
objectives in economics. 
In this environment, however, it is known that the strong fairness requirement of no envy, together with no wastage, implies efficiency \citep{Svensson83}. 
Our result strengthens this insight: when combined with strategy-proofness, individual rationality, and no subsidy, equal treatment of
equals---which is substantially weaker than no envy---together with no wastage 
implies efficiency.


While strategy-proofness, individual rationality and ETE are standard axioms
in mechanism design, we briefly discuss our remaining axioms.~\footnote{We do not consider
interim individual rationality and Bayesian incentive compatibility. This allows
us to work in a prior-free model.} 
No subsidy is satisfied by all standard auction formats.
It can also be justified as a requirement that prevents participation by ``fake'' bidders---agents who have (or report) very low valuations for the objects but seek to derive utility from subsidies.
The no wastage axiom can be motivated as a very weak form
of efficiency. 
It excludes auctions where reserve prices are kept.
In practice, we see many auctions without
a reserve price. 
Though reserve prices are used in many auctions too, often
sellers, such as governments, fail to commit to these reserve prices by reselling the unsold objects (spectrum licenses, for example). Hence, no-wastage
is a reasonable axiom in our model. 

Our result relies heavily on the fact that the auction designer believes that preferences of agents may exhibit income effect, i.e., non-quasilinear. While quasilinearity is a reasonable assumption when payments of agents are small, ignoring income effects in auctions involving large payments (for instance, in selling high-worth natural resources like mines, spectrum etc.) makes economics models unrealistic.
\citet{Bulow2017} mention budget-constraint among bidders (which induces non-quasilinear preferences) as one of the two main reasons why spectrum auctions are complex.
When preferences are quasilinear, there are many mechanisms other than the MWEP mechanism that satisfy our five properties \citep{Ryan16}.
Moreover, our result does not hold even when agents have specific non-quasilinear preferences.
\cite{Kazumura20} show that the MWEP mechanism is not the only mechanism that satisfies our properties if objects are normal goods.
Thus, a key assumption for our result is that the domain contains various preferences with income effects.\footnote{Another important assumption for our result is heterogeneity of objects.
If objects are identical, the MWEP mechanism is not the only mechanism that satisfies the five properties.
This is true both when preferences are quasilinear and when they can be non-quasilinear \citep{Adachi14}.}
While such domain requirement seems to be demanding, this is reasonable for a seller who has no information about the extent and nature of income effect in agents' preferences over transfers, wants to be robust about this aspect of the model.

We study a model with unit-demand agents, i.e., agents can be assigned at most one out of many objects. Though restrictive, this model appears in practice in many settings. The unit-demand assumption can be justified based on institutional restrictions. For instance, while selling team franchises in professional sports leagues, it is common to restrict a buyer to buy at most one franchise; the first spectrum auction in UK restricted each bidder to buy at most one spectrum \citep{Binmore02}. The unit-demand assumption can naturally arise from preferences of agents also. For instance, buyers in public housing markets are usually interested to buy at most one house \citep{Andersson16}.

\section{Related literature}
There is a growing literature on fair mechanisms in auction models. 
While many papers impose fairness requirements stronger than ETE, the class of mechanisms that satisfy ETE and incentive constraints is not well understood. 

When agents are unit-demand bidders, the only mechanism that satisfies no-envy, strategy-proofness, individual rationality, no subsidy, and no wastage is the MWEP mechanism \citep{Svensson83, Sakai13, Morimoto15}.
If objects are identical and preferences are quasilinear, the MWEP mechanism is characterized by anonymity, strategy-proofness, individual rationality, no subsidy, and no wastage \citep{Ashlagi12}.
However, when objects are heterogeneous and preferences are quasilinear, the MWEP mechanism is not the unique mechanism that satisfies anonymity and these properties \citep{Ryan16}.
This is also true when objects are identical and preferences are allowed to be non-quasilinear \citep{Adachi14}.
Since ETE is weaker than anonymity, there are multiple mechanisms that satisfy ETE, strategy-proofness, individual rationality, no subsidy, and no wastage in these environments.\footnote{\cite{Ryan16} and \cite{Adachi14} characterize the MWEP mechanism using a continuity axiom in addition to anonymity and other properties. }
Moreover, even if objects are identical and preferences are quasilinear, there are many mechanisms that satisfy these properties. 
Against this background, our result is remarkable in that the MWEP mechanism is uniquely pinned down as by these properties.

Several recent papers, like ours, impose fairness conditions as a constraint in addition to incentive constraints. 
But they pursue different objectives. 
\cite{Kazumura20} and \cite{Sakai21} show that the MWEP mechanism is ex-post revenue maximizing among mechanisms that satisfy ETE, strategy-proofness, individual rationality, no subsidy, and no wastage.
\cite{chen25} consider a single object model, and show that a second price auction with flexible reserve prices is expected revenue maximizing among mechanisms that satisfy anonymity and Bayesian incentive compatibility.  
While these papers adopt a similar approach to ours, they require revenue maximization as an additional axiom.
In contrast, our result shows that the MWEP mechanism is uniquely characterized without specifying any particular objective, once fairness is imposed as a constraint.

The Vickrey-Clarke-Groves (VCG) mechanism plays a central role when preferences are quasilinear.
It is a unique mechanism that satisfies efficiency, strategy-proofness, individual rationality, and no subsidy \citep{Holm79, Chew07}.
When preferences may not be quasilinear, however, the VCG is no longer efficient or strategy-proof.
In such settings with unit-demand agents, the MWEP mechanism is characterized using efficiency, strategy-proofness, individual rationality, and no subsidy \citep{Saitoh08, Sakai08, Morimoto15, Zhou18, Wakabayashi25}.

Our result parallels the characterization of the MWEP mechanism in \cite{Morimoto15}.
They consider the same model and domain and characterize the MWEP using efficiency.
Since no-wastage is a significant weakening of Pareto efficiency, our result highlights that equity can almost substitute efficiency objectives of an auction designer. 
We believe that such an alternate foundation is useful because
fairness is an important desideratum in auction design and its relation to efficiency is not well understood in auction models.
Besides, equal treatment of equals is a more \emph{testable} axiom in practice than Pareto efficiency.
For instance, there are
legal implications of violating equal treatment of equals in
auctions -- \citet{Deb16} contain examples of some prominent lawsuits in the United States where auctioneers have been dragged to court for designing auction mechanisms
that discriminate among bidders. 
On the other hand, violating Pareto efficiency usually does not have any legal implications.

In models with multi-demand agents, a Walrasian equilibrium may not exist.\footnote{There is a large literature on the existence of a Walrasian equilibrium; see, for example, \cite{Kelso82}, \cite{Sun06}, \cite{Teytelboym14}, and \cite{Baldwin19}. 
The existence with non-quasilinear preferences is studied by
\cite{Baldwin23} and \cite{Nguyen24}.}
Even when a Walrasian equilibrium exists, the MWEP mechanism may not be strategy-proof.
Moreover, when preferences may not be quasilinear, an efficient and strategy-proof mechanism does not exist \citep{Kazumura16,
Baisa20,Malik19, Shinozaki25}.
Thus, \cite{Shinozaki25b} pursue constrained efficiency, and introduce the notion of the bundling unit-demand minimum price Walrasian mechanism. 
This mechanism partitions objects into several bundles and selects a minimum price Walrasian equilibrium when the bundles are regarded as objects. 
They show that this is the unique mechanism that satisfies constrained efficiency, ETE, strategy-proofness, and no subsidy.


\section{Model and definitions}

There are $n\ge  2$ agents and $m\ge  2$ objects with $n > m$. We denote the set of
agents by $N\equiv \{1,\dots ,n\}$ and the set of objects by $M\equiv
\{1,\dots ,m\}$. Let $L\equiv M\cup \{0\}$, where assigning $0$ means not assigning
any (real) object from $M$. We call $0$ the null object, and unlike a real object in
$M$, it can be assigned to any number of agents. Each agent receives at most one object
and pays some amount of money. Thus, agents' common {\bf consumption set}
is $L \times \mathbb{R}$, and a generic {\bf (consumption) bundle}
for any agent $i\in N$ is a pair $z_i=(a,t)\in L \times \mathbb{R}$.

\subsection{Classical preferences}

Each agent $i\in N$ has a complete and transitive preference $R_i$
over $L\times \mathbb{R}$. Let $P_i$ and $I_i$ be the strict and
indifference relations associated with $R_i$.
\begin{defn}
A preference $R_i$ is {\bf classical} if it satisfies the following four conditions:
\begin{enumerate}

\item {\bf Money monotonicity.} For every $t > t'$ and for every $a \in L$,
we have $(a,t')~P_i~(a,t)$.

\item {\bf Desirability of objects.} For every $t$ and for every $a \in M$, we have $(a,t)~P_i~(0,t)$.

\item {\bf Continuity.} For every $z \in L \times \mathbb{R}$, the sets $\{z': z'~R_i~z\}$ and $\{z': z~R_i~z'\}$ are closed.

\item {\bf Possibility of compensation.} For every $z \in L \times \mathbb{R}$ and
for every $a \in L$, there exists $t$ and $t'$ such that $z~R_i~(a,t)$ and $(a,t')~R_i~z.$
\end{enumerate}
\end{defn}

Let $\mathcal{R}^C$ be the set of all classical preferences. We call $(\mathcal{R}^C)^n$ the {\bf classical domain}. 
Throughout this paper, we assume the preferences to be classical.
We use $\mathcal{R} \subseteq \mathcal{R}^C$ to denote an arbitrary domain of preferences.

A preference $R_i$ is {\bf quasilinear} if there exists a {\bf valuation function} $v: L \rightarrow \mathbb{R}_+$
with $v(0)=0$ such that for all $a,b \in L$ and for all $t,t' \in \mathbb{R}$,
we have $(a,t)~R_i~(b,t')$ if and only if $v(a)-t \ge v(b) - t'$.
We denote the set of all quasilinear preferences as $\mathcal{R}^Q$.

 

The existence of a valuation function $v:L\rightarrow \mathbb{R}$ is restricted to quasilinear preferences. 
However, the notion of valuation can be extended to classical preferences.
\begin{defn}
The {\bf valuation} of agent $i$ with preference $R_i$ for object $a \in L$ at consumption bundle
$z$ is defined as $V^{R_i}(a,z)$, which uniquely solves $(a,V^{R_i}(a,z))~I_i~z$.
\end{defn}

In other words, $V^{R_i}(a,z)$ is the unique amount of transfer that makes agent $i$ indifferent
between consumption bundles $z$ and $(a,V^{R_i}(a,z))$. The existence of $V^{R_i}(a,z)$ and its
uniqueness are guaranteed by the assumptions of classical preferences.\footnote{See \cite{Kazumura16} for the formal proof.} 
Note that if $R_i$ is quasilinear, then for each $a\in L$, $V^{R_i}(a, (0, 0))=v_i(a)$.

\subsection{Mechanisms}

An {\bf object allocation} is an $n$-tuple $(a_1,\dots ,a_n)\in
L^n$ such that for each pair $i,j\in N$ with $i\neq j$, $a_i= a_j$ implies $a_i=a_j=0$.
We denote the set of object allocations by $A$.
A {\bf (feasible) allocation} is an $n$-tuple $z\equiv (z_1,\dots ,z_n)\equiv ((a_1,t_1),\dots ,(a_n,t_n))\in (L\times \mathbb{R})^n$ such that $(a_1,\dots ,a_n)\in A$.
We denote the set of feasible allocations by $Z$.

Fix a domain $\mathcal{R}^n \subseteq (\mathcal{R}^C)^n$.
A {\bf preference profile} is an $n$-tuple $R\equiv (R_1,\dots, R_n)\in
\mathcal{R}^n$. Given $R\in \mathcal{R}^n$ and $i\in N$, let $R_{-i}\equiv
(R_j)_{j\neq i}$. Given $R\in \mathcal{R}^n$ and $N'\subseteq N$,
let $R_{N'}=(R_i)_{i\in N'}$ and $R_{-N}\equiv (R_i)_{i\in
N\setminus N'}$.

A {\bf mechanism} on $\mathcal{R}^n$ is a function $f:\mathcal{R}^n\rightarrow Z$. Given a mechanism $f$
and $R\in \mathcal{R}^n$, we denote the bundle assigned to agent $i$ by $
f_i(R)$ and we write $f_i(R)=(a_i(R),t_i(R))$, where $a_i(R)$ is the object assigned to agent $i$ and $t_i(R)$ is his payment. We require a mechanism to satisfy the following properties.

\begin{defn}
Let $f:\mathcal{R}^n \rightarrow Z$ be a mechanism defined on domain $\mathcal{R}$. 
\begin{enumerate}

\item $f$ is {\bf strategy-proof} if for every $i \in N$,
for every $R_{-i} \in \mathcal{R}^{n-1}$, and for every $R_i,R'_i \in \mathcal{R}$, we have
$$f_i(R_i,R_{-i})~R_i~f_i(R'_i,R_{-i}).$$

\item $f$ is {\bf (ex-post) individual rationality (IR)} if for every $i \in N$, for every $R \in \mathcal{R}^n$,
we have $f_i(R)~R_i~(0,0).$

\item $f$ satisfies {\bf equal treatment of equals (ETE)} if for every $i,j \in N$, for every $R \in \mathcal{R}^n$ with $R_i=R_j$,
we have $f_i(R)~I_i~f_j(R).$

\item $f$ satisfies {\bf no wastage (NW)} if for every $R \in \mathcal{R}^n$ and for every $a \in M$, there
exists some $i \in N$ such that $a_i(R)=a$.

\item $f$ satisfies {\bf no subsidy (NS)}
if for each $R\in \mathcal{R}^n$ and each $i\in N$, we have $t_i(R)\ge  0$.

\end{enumerate}

\end{defn}

Earlier in Section \ref{sec:intro}, we discussed the above properties in
detail. They can either be motivated by weak fairness (ETE) or weak efficiency (NW) or
some practical concerns (NS). Next, we introduce the MWEP mechanism.

\section{A characterization of the MWEP mechanism}

In this section, we define the minimum Walrasian equilibrium price mechanism.
A price vector $p \in \mathbb{R}^{|L|}_+$
defines a price for every object with $p_0=0$. At any price vector $p$,
let $D(R_i,p)\equiv \{a \in L: (a,p_a)~R_i~(b,p_b)~\forall~b \in L\}$ denote
the {\bf demand set} of agent $i$ with preference $R_i$ at price vector $p$.

\begin{defn}
An object allocation $(a_1,\ldots,a_n) \in A$ and a price vector $p$ is a {\bf Walrasian equilibrium}
at a preference profile $R \in (\mathcal{R}^C)^n$ if
\begin{enumerate}
\item $a_i \in D(R_i,p)$ for all $i \in N$ and
\item for all $a \in M$ with $a \ne a_i$ for all $i \in N$, we have $p_a=0$.
\end{enumerate}
We refer to $p$ and $\{z_i \equiv (a_i,p_{a_i})\}_{i \in N}$ defined above as a {\bf Walrasian equilibrium price vector}
and a {\bf Walrasian equilibrium allocation} at $R$ respectively.
\end{defn}

It is known that a Walrasian equilibrium exists at each $R\in (\mathcal{R}^C)^n$ (\cite{Demange85, Alkan90}).
A Walrasian equilibrium price vector $p$ is a {\bf minimum Walrasian equilibrium price vector} at preference
profile $R$ if for every Walrasian equilibrium price vector $p'$ at $R$, we have
$p_a \le p'_a$ for all $a \in L$. \citet{Demange85} prove that if $R$ is a profile of classical preferences, then
the set of Walrasian equilibrium price vectors at $R$ forms a lattice with a unique minimum and a unique maximum.
We denote the minimum Walrasian equilibrium price vector at $R$ as $p^{\min}(R)$. Note that for each $R \in \mathcal{R}^n$, even though $p^{\min}(R)$ is a unique
price vector, there may be many object allocations that can support the Walrasian equilibrium.
Let
$Z^{\min}(R)$ denote the set of all allocations at a minimum Walrasian equilibrium at preference
profile $R$. Note that if $((a_i,t_i))_{i\in N} \in Z^{\min}(R)$ then $t_i=p^{\min}_{a_i}(R)$, i.e.,
the transfer associated with an agent is the price of the object assigned to him in the
Walrasian equilibrium. 
\begin{defn}
A mechanism $f:\mathcal{R}^n \rightarrow Z$ is a {\bf minimum Walrasian equilibrium (MWEP) mechanism}
if $$f(R) \in Z^{\min}(R)~\forall~R \in \mathcal{R}^n.$$
\end{defn}

Let $f^{\min}: \mathcal{R}^n \rightarrow Z$ denote an MWEP mechanism.
For every $R \in \mathcal{R}$, an MWEP mechanism picks {\em any}
one allocation in $Z^{\min}(R)$. It can be easily shown that each agent must be indifferent between
all its allocations in $Z^{\min}(R)$. Hence, whenever we say {\em the} MWEP mechanism,
we mean {\em any} MWEP mechanism.

\citet{Demange85} showed that the MWEP mechanism is strategy-proof and IR. Clearly, it also satisfies
ETE, NW, and NS. Our main result is that on the classical domain, only the MWEP mechanism satisfies these properties.
\begin{theorem}\label{theo:main}
Let $f: (\mathcal{R}^C)^n \rightarrow Z$ be a mechanism. Then, $f$
is strategy-proof, individually rational, and satisfies equal treatment of equals, no wastage, and no subsidy if and only if it is the MWEP mechanism.
\end{theorem}

We postpone the proof of Theorem \ref{theo:main} to the appendix.
We argue that the domain richness in Theorem \ref{theo:main}
is somewhat necessary. \citet{Ryan16} contains an
example which shows that there are  strategy-proof and IR mechanisms satisfying ETE, NW, and NS that are not MWEP mechanisms if the domain of preferences is domain of quasilinear preferences.

Similarly, the axioms in Theorem \ref{theo:main} are necessary. Below, we
provide some examples to establish this claim.
\begin{itemize}
\item \textsc{Strategy-proofness}.
Consider a mechanism that chooses the maximum Walrasian equilibrium allocation at every profile.
Such a mechanism will satisfy all the properties except strategy-proofness \citep{Miyake98}.

\item \textsc{Individual rationality}.
Consider the MWEP mechanism supplemented by a constant participation fee (which is independent of
the preferences of the agents and equal across all the agents).
Such a mechanism will satisfy all the properties except IR.

\item \textsc{Equal treatment of equals}. Consider the following mechanism, where we
treat some agent, say agent $1$, differently. At every preference profile,
agent $1$ is asked to pick her best bundle in $Z_0\equiv \{(a,0): a \in M\}$, i.e., she
picks her best object at zero payment. Then, for the remaining agents and remaining
objects, we use the MWEP mechanism for the subeconomy. Such a mechanism is clearly strategy-proof,
individually rational, and satisfies no wastage and no subsidy. However, it fails
equal treatment of equals.

\item \textsc{No wastage}.
Consider the mechanism which never sells any of the objects and does not ask agents to pay anything.
This mechanism satisfies all the properties except
no wastage.

\item \textsc{No subsidy}.
Consider the MWEP mechanism supplemented by a constant participation subsidy (which is independent
of the preferences of the agents and equal across all the agents).
Such a mechanism satisfies all the properties except no subsidy.

\end{itemize}

\section{Connection to earlier characterizations}

We describe the connection of our characterization to other characterizations
of the MWEP mechanism in the literature. Most of these characterizations involve
Pareto efficiency. 

\begin{defn}
A mechanism $f:\mathcal{R}^n \rightarrow Z$ is {\bf Pareto efficient} if at
every preference profile $R \in \mathcal{R}^n$, there exists no allocation $%
((\hat{a}_1,\hat{t}_1),\ldots,(\hat{a}_n,\hat{t}_n))\in Z$ such that
\begin{align*}
(\hat{a}_i,\hat{t}_i)~&R_i~f_i(R)~\qquad~\forall~i \in N, \\
\sum_{i \in N}\hat{t}_i &\ge \sum_{i \in N}t_i(R),
\end{align*}
with either the second inequality holding strictly or some agent $i$
strictly preferring $(\hat{a}_i,\hat{t}_i)$ to $f_i(R)$.
\end{defn}

The above definition is the appropriate notion of Pareto efficiency in this
setting. Notice that by distributing some money among all the agents, we can always
make each agent better off than the allocation in any mechanism. Hence, the above definition requires that there should
not exist another allocation where the sum of transfers
is not less and every agent is weakly better off.

The MWEP mechanism is Pareto efficient. The following theorem
characterizes the MWEP mechanism using Pareto efficiency. We remind that $\mathcal{R}^Q$
denotes the set of quasilinear preferences.

\begin{theorem}[\citet{Holm79,Morimoto15}]\label{theo:eff}
Suppose $\mathcal{R} \in \{\mathcal{R}^Q, \mathcal{R}^C\}$ and
let $f:\mathcal{R}^n \rightarrow Z$ be a mechanism on this domain.
Then, $f$ is strategy-proof, Pareto efficient, individually rational,
and satisfies no subsidy if and only if it is an MWEP mechanism.
\end{theorem}

Theorem \ref{theo:eff} was proved for $\mathcal{R} = \mathcal{R}^Q$
by \citet{Holm79} and for $\mathcal{R} = \mathcal{R}^C$
by \citet{Morimoto15}.~\footnote{\citet{Zhou18} have shown that
Theorem \ref{theo:eff} continues to hold in smaller non-quasilinear preference
domains.}

The efficiency characterization of Theorem \ref{theo:eff} holds
for both the quasilinear domain and the classical domain. This is not true for
our characterization in Theorem \ref{theo:main}.
As discussed later, our characterization holds for the classical
domain but breaks down for the quasilinear domain. We can use both the
characterizations to point out an interesting connection between Pareto
efficiency and ETE with no wastage in the classical domain.\
\begin{cor}\label{cor:eff}
  Suppose $f:(\mathcal{R}^C)^n \rightarrow Z$ is a strategy-proof and individually
  rational mechanism satisfying no subsidy. Then, the following are equivalent.

  \begin{enumerate}

\item $f$ is Pareto efficient.

\item $f$ satisfies equal treatment of equals and no wastage.

  \end{enumerate}

\end{cor}
Corollary \ref{cor:eff} is true because in the classical domains, the only mechanism
satisfying these axioms is the MWEP mechanism, which is not the case in the quasilinear
domain.
The connection between fairness and efficiency has been known in the literature. \cite{Svensson83} shows that no envy and no wastage imply Pareto efficiency.\footnote{A mechanism $f:\mathcal{R}^n\rightarrow Z$ satisfies {\it no envy} if for each $R\in \mathcal{R}^n$ and each pair $i, j\in N$, $f_i(R)~R_i~f_j(R)$. It is clear that no envy implies equal treatment of equals.}
Corollary \ref{cor:eff} shows that equal treatment of equal, which is substantially weaker than no envy, and no wastage imply Pareto efficiency under additional auxiliary axioms such as strategy-proofness and individual rationality.

The other result which is worth explaining is \citet{Kazumura20}.
They consider a class of domains that they call {\em rich} domains.
Formally, it is defined as follows.
\begin{defn}
A domain of preference $\mathcal{R}$ is {\bf rich} if for all
$a \in M$ and for all $\hat{p}$ with $\hat{p}_a > 0$ and $\hat{p}_b = 0$
for all $b \ne a$ and for every $p$ with $p_x > \hat{p}_x$ for all $x \in M$,
there exists a preference
$R_i \in \mathcal{R}$ such that $D(R_i,\hat{p})=\{a\}$ and $D(R_i,p)=\{0\}$.
\end{defn}



It is not difficult to see that many domains of preferences can be rich.
\citet{Kazumura20} show that the domain of quasilinear preferences ($\mathcal{R}^Q$),
the domain of classical preferences ($\mathcal{R}^C$), and domains
containing all {\em positive income effect} preferences satisfy richness.
If a domain is rich, then the following result holds.
Denote the {\em revenue} at a preference profile $R$ from a mechanism $f$
as $\textsc{Rev}^f(R)\equiv \sum_{i \in N}t_i(R)$, i.e., the sum of
payments of all the agents at preference profile $R$.
\begin{defn}
A mechanism $f:\mathcal{R}^n \rightarrow Z$ is {\bf ex-post revenue optimal}
among a class of mechanisms if for every mechanism $g:\mathcal{R}^n \rightarrow Z$
in this class,
\begin{align*}
\textsc{Rev}^f(R) \ge \textsc{Rev}^g(R)~\qquad~\forall~R \in \mathcal{R}^n.
\end{align*}
\end{defn}
\begin{theorem}[\citet{Kazumura20}]\label{theo:rich}
Suppose $\mathcal{R}$ is a rich domain. Then, every MWEP mechanism is
ex-post revenue optimal among the class of strategy-proof and IR mechanisms satisfying equal treatment of equals, no wastage, and no subsidy.
\end{theorem}

Since the classical domain is rich, it follows that Theorem \ref{theo:rich} is an easy corollary
of our main result Theorem \ref{theo:main} when applied to the classical domain. However, such a conclusion cannot be
drawn for other rich domains. For instance, 
the quasilinear domain of preferences is
rich, and we know that Theorem \ref{theo:rich} holds but Theorem \ref{theo:main} does not hold in that domain.

Two other characterizations are worth pointing. \citet{Sakai13} shows
that if there is a {\em single} object, then the MWEP mechanism
is the unique strategy-proof and individually rational mechanism satisfying ETE, NW, and NS in the classical domain.
Hence, we generalize the result of \citet{Sakai13} to the case of multiple
heterogeneous objects. \citet{Ashlagi12} characterizes the MWEP mechanism
for the {\em multiple identical} objects case in the {\em quasilinear} domain,
when agents can be assigned at most one object,
using strategy-proofness, individual rationality, no subsidy,
no wastage, and {\em anonymity} (which is stronger than equal treatment of equals).
As we have discussed, our characterization does not hold in the quasilinear domain.
Further, the result of \citet{Ashlagi12} does not extend to the case where
objects are {\em heterogenous} as \citet{Ryan16} shows that there
are strategy-proof and individually rational mechanisms satisfying no subsidy,
no wastage, and anonymity which is not the MWEP mechanism.
Also, the result of \citet{Ashlagi12} does not extend to the case where objects are identical but preferences are classical \citep{Adachi14}.

\section{Outline of the proof of Theorem \ref{theo:main}}
In this section, we first introduce a useful representation of preferences, which we often use in the proof of Theorem \ref{theo:main}.
We then explain the sketch of the proof using a simple example.
\subsection{Indifference vectors}
\begin{figure}[!hbt]
\centering
\includegraphics[height=2.in]{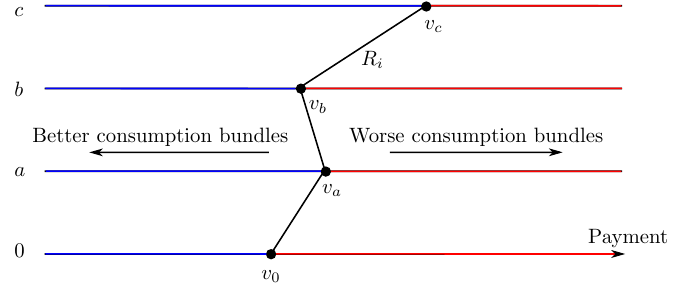}
\caption{An indifference vector}
\label{fig:indiff}
\end{figure}
A vector $v \in \mathbb{R}^{m+1}$ is an {\bf indifference vector} of classical preference
$R_i$ if for all $a, b \in L$, we have $(a,v_a)~I_i~(b,v_b)$. Denote the set of all
indifference vectors of $R_i$ as $\mathcal{I}(R_i)$. 
A typical indifference vector $v$ of a preference $R_i$ can be represented
by a diagram shown in Figure \ref{fig:indiff} for three objects (and the null object).
Each horizontal line
in the figure corresponds to a unique object. Each point on each of the horizontal lines corresponds
to a payment level. So, the set of all consumption bundles
are the four horizontal lines in Figure \ref{fig:indiff}.
As we go right along the horizontal lines, the payment of the agent increases. Hence, consumption bundles
to the right (left) of the indifference vector $v$ shown in Figure \ref{fig:indiff} are worse (respectively, better) than the four
consumption bundles corresponding to $v$.

An equivalent way to think of a preference $R_i$ is through its indifference vectors $\mathcal{I}(R_i)$, which
is an infinite set. Hence, a preference consists of an infinite collection of such vectors: an
illustration is shown in Figure \ref{fig:pref}.
\begin{figure}[!hbt]
\centering
\includegraphics[height=2.in]{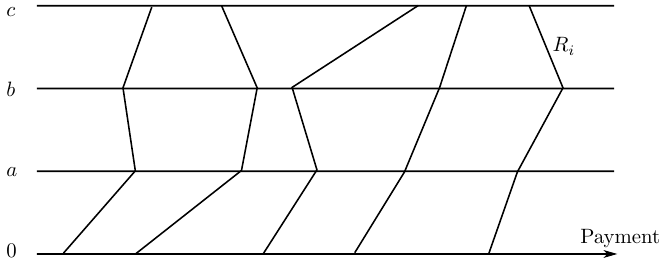}
\caption{A preference and its indifference vectors}
\label{fig:pref}
\end{figure}
Note that for every classical preference $R_i$ and for every distinct $v, v' \in \mathcal{I}(R_i)$, we have
either $v > v'$ or $v' > v$, i.e., $v$ and $v'$ do not intersect.

\subsection{Outline of the proof}\label{Outline}
We now provide an outline of the proof of Theorem~\ref{theo:main}.
For brevity of exposition, we refer to any strategy-proof and individually rational mechanism satisfying equal treatment of equals, no wastage and no subsidy as a {\bf desirable mechanism}.

Since the MWEP mechanism is desirable, what we need to show is that
there is no other mechanism that is desirable.
\cite{Kazumura20} shows that if a mechanism is desirable, then for each preference profile, the mechanism assigns an allocation where each agent receives a bundle that is at least as desirable as bundles that he receives at minimum price Walrasian equilibrium allocations.
\begin{fact}[\citet{Kazumura20}]\label{domination}
Let $f$ be a desirable mechanism on $\mathcal{R}^n$.
For each $R\in \mathcal{R}^n$, each $((a_i, t_i))_{i\in N}\in Z^{\min}(R)$, and each $i\in N$, $f_i(R)~R_i~(a_i, t_i)$.
\end{fact}

Let $\mathcal{R}=\mathcal{R}^C$ and $f: \mathcal{R}^n \rightarrow Z$ be a desirable mechanism.
By Fact~\ref{domination}, to complete the proof, it is sufficient to show that for each $R\in \mathcal{R}^n$ and each $i\in N$, $t_i(R)= p^{\min}_{a_{i(R)}}(R)$.\footnote{
Let $R\in \mathcal{R}^n$ and suppose that for each $i\in N$, $t_i(R)=p^{\min}_{a_i(R)}(R)$. Let $((a_j, t_j))_{j\in N}\in Z^{\min}(R)$ and $i\in N$. By Fact~\ref{domination}, $(x_i(R), p^{\min}_{a_i(R)}(R))=f_i(R)~R_i~(a_i, t_i)$. Thus, by $a_i\in D(R_i, p^{\min}(R))$ and thus, $x_i(R)\in D(R_i, p^{\min}(R))$. Hence $f(R)\in Z^{\min}(R)$.}
Fact~\ref{domination} also implies that for each $R\in \mathcal{R}^n$ and each $i\in N$, $t_i(R)\le p^{\min}_{a_{i(R)}}(R)$.\footnote{To see this, let $R\in \mathcal{R}^n$, $((a_i, t_i))_{i\in N}\in Z^{\min}(R)$, and $i\in N$. By Fact~\ref{domination} and $a_i\in D(R_i, p)$, $f_i(R)~R_i~(a_i, t_i)~R_i~(x_i(R), p^{\min}_{a_i(R)}(R))$.
This implies $t_i(R)\le p^{\min}_{a_{i(R)}}(R)$.}
Therefore, all we need to show is that for each $R\in \mathcal{R}^n$ and each $i\in N$, $t_i(R)\ge p^{\min}_{a_{i(R)}}(R)$.
$$
t_i(R)\ge p^{\min}_{a_i(R)}(R) \text{ for each } R\in \mathcal{R}^n \text{ and each }i\in N.
$$

In this section, we overview the proof by focusing on a simple example.
Suppose there are three agents and two objects.
For convenience, denote $M=\{x_1, x_2\}$.
Let $v\in \mathbb{R}_{++}$.
We focus on a preference profile $R\in \mathcal{R}^3$ such that for each $i\in N$,
$$
V^{R_i}(x_1, (0, 0))=V^{R_i}(x_2, (0, 0))= v.
$$

Note that if we assume preferences to be quasilinear, this condition implies preferences in $R$ are identical.
In this case, since at least one agent receives $(0, 0)$, {\sl equal treatment of equals} immediately implies that for each $i\in N$, $f_i(R)~I_i~(0, 0)~I_i~(x_1, p^{\min}_{x_1}(R))~I_i~(x_2, p^{\min}_{x_2}(R))$, completing the proof.
However, we cannot make this assumption, and some of the preferences in $R$ may not be identical.
Here, we focus on the case where preferences in $R$ are all distinct.

Suppose for contradiction that there is an agent, say agent $1$, such that $t_1(R)<p^{\min}_{a_1(R)}(R)$.
By no subsidy, $a_1(R)\ne 0$.
Without loss of generality, assume $a_1(R)=x_1$.
It is easy to see that $p^{\min}_{x_1}(R)=p^{\min}_{x_2}(R)=v$.
Thus, $t_1(R)<p^{\min}_{x_1}(R)=v$.



The main difficulty of our proof is that {\sl equal treatment of equals} can be used only when some agents have the same preference.
Indeed, since preferences are all distinct at $R$, {\sl equal treatment of equals} has no implication for the allocation at $R$.
Therefore, by changing preferences from $R$, we have to construct a preference profile where some agents have identical preferences and {\sl equal treatment of equals} induces a contradiction.
A significant part of the proof is dedicated to construct such a preference profile. We give a sense of this construction below.

We first change agent~1's preference.
Let $R'_1\in \mathcal{R}$ satisfy the following properties:
\begin{enumerate}
\item $V^{R'_1}(x_2, f_1(R))<0$.
\item For each $x\in M$,
$$
V^{R'_1}(x, (0, 0))
\begin{cases}
<t_1(R)+\epsilon&\text{if }x=x_1,
\\
>v&\text{if }x=x_2,
\end{cases}
$$
where $\epsilon >0$ is sufficiently close to zero so that $V^{R'_1}(x_1, (0, 0))<v$. 
\end{enumerate}
\begin{figure}[t!]
  \centering
\includegraphics[height=2.in]{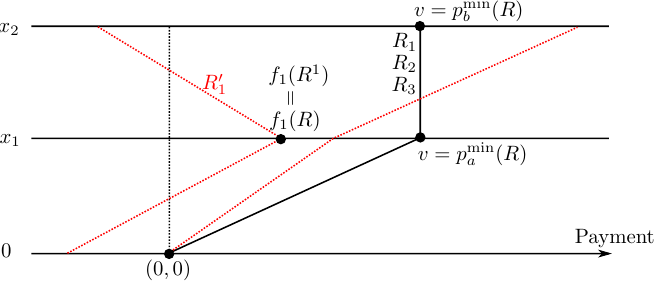}
\caption{An illustration of $R'_1$.}
\label{fig:R'_1}
\end{figure}

Figure~\ref{fig:R'_1} is an illustration of $R'_1$.
The interpretation of $R'_1$ is as follows: The first condition means that, given that agent~1 receives $f_1(R)$, $x_1$ is more preferred to $x_2$ in the sense that agent~1 would not give up $x_1$ in exchange for $x_2$ unless she receives a positive amount of money.
In this sense, agent~1 ``favors'' $f_1(R)$ -- in the proof, we call such a preference a $f_1(R)$-favoring preference (see Appendix~\ref{Preliminaries}).
The second condition means that at $(0, 0)$, $x_2$ is (much) more preferred and $x_1$ is (slightly) less preferred under $R'_1$ than under $R_1$.
This condition also implies that at $(0, 0)$, $x_2$ is preferred to $x_1$ in the sense that the valuation for $x_2$ at $(0, 0)$ is higher than that for $x_1$ at $(0, 0)$.
Since $x_1$ is preferred to $x_2$ at $f_1(R)$, $R'_1$ exhibits income effects.
In the formal proof, we need to define $R'_1$ to satisfy more delicate conditions to incorporate issues that arise in a general setting.

The implication of this delicate construction is the following. Denote $R^1:=(R'_1, R_2, R_3)$.
First, we see that $f_1(R^1)=f_1(R)$.
By strategy-proofness, $f_1(R^1)~R'_1~f_1(R)$.
Thus, if $a_1(R^1)\ne x_1$, then by the construction of $R'_1$, $t_1(R^1)<0$.
This contradicts {\sl no subsidy}.
Thus, $a_1(R^1)=x_1$, and then, strategy-proofness immediately implies $f_1(R^1)=f_1(R)$.
Next, by no wastage, there is an agent who receives $x_2$ at $R^1$.
As we have seen, this is not agent~1.
Without loss of generality, assume that agent~2 receives $x_2$.

Now, we change agent~2's preference.
The idea is similar to $R'_1$, but the roles of $x_1$ and $x_2$ are interchanged.
\begin{figure}[t!]
  \centering
\includegraphics[height=2.in]{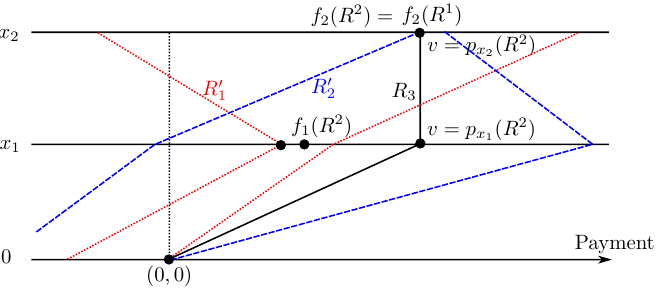}
\caption{An illustration of $R'_2$.}
\label{fig:R'_2}
\end{figure}
Let $R'_2\in \mathcal{R}$ satisfy the following properties:
\begin{enumerate}
\item $V^{R'_2}(x_1, f_2(R^1))<0$.
\item For each $x\in M$,
$$
V^{R'_2}(x, (0, 0))
\begin{cases}
>v&\text{if }x=x_1,
\\
<t_2(R^1)+\epsilon&\text{if }x=x_2,
\end{cases}
$$
where $\epsilon >0$ is sufficiently close to zero so that if $t_2(R^1)<v$, then $V^{R'_2}(x_2, (0, 0))<v$.
\end{enumerate}

Again, here are the implications of the delicate construction of $R'_2$. See Figure~\ref{fig:R'_2} for an illustration of $R'_2$.
Denote $R^2=(R'_1, R'_2, R_3)$.
Note that as in the case of $R^1$, strategy-proofness implies that $f_2(R^2)=f_2(R^1)$.
Then, $a_1(R^2)=x_1$ or $0$.
Suppose $a_1(R^2)=0$.
Individual rationality and no subsidy imply $t_1(R^2)=0$.
It is easy to show that $p^{\min}_{x_1}(R^2)=p^{\min}_{x_2}(R^2)=v$ (see Figure~\ref{fig:R'_2}).
By the construction of $R'_1$, we have $(x_2,p^{\min}_{x_2}(R^2)) ~P'_1~(0, 0)=f_1(R^2)$.
This contradicts Fact~\ref{domination}.
Hence, $a_1(R^2)=x_1$.
By $V^{R'_1}(x_1, (0, 0))<v$ and individual rationality, $t_1(R^2)<v$.

\begin{figure}[t!]
  \centering
\includegraphics[height=2.in]{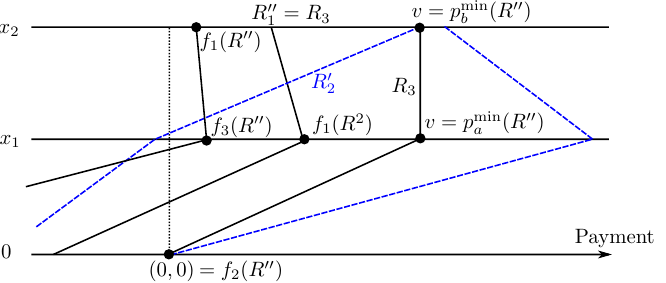}
\caption{An illustration of $R''$.}
\label{fig:R''_1}
\end{figure}
Note that by $t_1(R^2)<v$, $f_1(R^2)~P_3~(0, 0)$.
This is, in fact, a key condition to derive a contradiction.
Now, let $R''_1\in \mathcal{R}$ be such that $R''_1=R_3$, and denote $R'':=(R''_1, R'_2, R_3)$.
For an illustration of $R''$, see Figure~\ref{fig:R''_1}.
By strategy-proofness, $f_1(R'')~R''_1~f_1(R^2)$.
Thus, by $f_1(R^2)~P_3~(0, 0)$ and $R''_1=R_3$, $f_1(R'')~P''_1~(0, 0)$.
By no subsidy, this implies $a_1(R'')\ne 0$.
Since agents~1 and 3 have the same preference, $f_3(R'')~I_3~f_1(R'')~P_3~(0, 0)$.
Hence, $a_3(R'')\ne 0$.

Since there are only two objects, agent~2 receives $0$.
By individual rationality and no subsidy, $f_2(R'')=(0,0)$.
It is also easy to show that $p^{\min}_{x_1}(R'')=p^{\min}_{x_2}(R^*)=v$.
However, by the construction of $R'_2$, we have $(x_1, p^{\min}_{x_1}(R''))~P'_2~(0, 0)=f_2(R'')$.
This contradicts Fact~\ref{domination} and we complete the proof.

In the formal proof, we need to change preferences of $m$ agents, where $m$ is the number of objects, and construct $m$ preference profiles:
\begin{itemize}
\item $R^1:=(R'_1, R_2,\dots, R_n)$,
\item $R^2:=(R'_1, R'_2, R_3,\dots, R_n)$,\vspace{0.1cm}
\\ \vspace{-0.22cm} $\vdots$
\item $R^m:=(R'_1, \dots, R'_m, R_{m+1},\dots, R_n)$.
\end{itemize}
Then, we show that, at $R^m$, there are $i\in \{1,\dots, m\}$ and $j\in N\setminus \{1,\dots, m\}$ such that $f_i(R^m)~P_j~(0, 0)$.
In the above example, it was easy to find such a pair of agents, as the number of agents who prefer $f_1(R)$ to $(0, 0)$ exceeds the number of objects.
However, this may not be the case in general.
In fact, the proof for the existence of such a pair is the most difficult part in our proof. Out of eleven steps in our proof, ten steps are dedicated to show the existence of such a pair of agents.

\section{Conclusion}
\label{sec:conc}
The proof of Theorem \ref{theo:main} is tedious and long, and unfortunately, there is no intuition to explain the result. The proof involves carefully constructing classical preferences and putting together the implications of our axioms at these preferences. While we wished a simpler proof was available for our theorem, it seems unlikely. At the same time, we believe that our result is useful. It provides an equity foundation (using equal treatment of equals) of the MWEP mechanism instead of the efficiency foundation (using Pareto efficiency) that is well-known. It answers the following question for the classical domain of preferences: {\sl Which strategy-proof and individually rational mechanisms satisfy no wastage, no subsidy, and equal treatment of equals?} Our result contributes to the literature on auctions with income effect. In future, we plan to explore an answer to this question in restricted domains of preferences like the quasilinear domain or the domain of positive income effect preferences.

\appendix

\section{Proof of Theorem \ref{theo:main}}

\subsection{Preliminaries}\label{Preliminaries}
First, we state three useful facts. 
The first fact states that if $n>m$, then the minimum Walrasian equilibrium price vector is always positive.

\begin{fact}
\label{positive}
Suppose $n>m$.
For each $R\in \mathcal{R}^n$ and each $a\in M$, $p_a^{\min}(R)>0$.
\end{fact}

The following notions play an important role in the proof.

\begin{defn}
Let $R\in \mathcal{R}^n$ and $p\in \mathbb{R}^{|L|}_+$ be a price vector.
A set of real objects $M'\subseteq M$ is {\bf overdemanded at $p$ for $
R$} if
\begin{align*}
|\{i\in N:D(R_i,p)\subseteq M'\}| >|M'|.
\end{align*}

A set of real objects $M'\subseteq M$ is {\bf underdemanded at $p$ for $
R$} if $p_a > 0$ for all $a \in M'$ and
\begin{align*}
|\{i\in N: D(R_i,p)\cap
M'\neq \emptyset\}| < |M'|,
\end{align*}
and {\bf weakly underdemanded at $p$ for $R$} if the above inequality holds
weakly.
\end{defn}

The following fact is a characterization of the minimum Walrasian
equilibrium price vector by means of overdemanded and weakly underdemanded
sets.

\begin{fact}[\citet{Mishra10, Morimoto15}]
\label{overdemand} Let $R\in \mathcal{R}^n$ and $p\in \mathbb{R}^{|L|}_+$ be a price vector.
Then, $p$ is a minimum Walrasian equilibrium price vector at $R$ if and only if
no set of real objects is overdemanded and no set of real objects is weakly underdemanded at $p$ for $R$.
\end{fact}

Fact~\ref{overdemand} implies the following fact.

\begin{fact}[Demand connected sequence]
\label{path} Let $n>m$, $R\in \mathcal{R}^n$, and $((a_i, t_i))_{i\in N}\in Z^{\min}(R)$. Then, for every agent $i^* \in N$, there
is a sequence of $K$ distinct agents $\{i_k\}_{k=1}^K$ such that\footnote{If $a_{i^*}=0$, then the sequence is $\{i_k\}_{k=1}^K=\{i^*\}$ and thus the latter part of Condition (2) and Condition (3) vacuously hold.} \\
(1) $i_1=i^*$,\\
(2) $a_{i_K}=0$ and  for each $k\in \{1,\dots, K-1\},\ a_{i_k}\ne 0$, and \\
(3) for each $k\in \{2,\dots ,K\}$, $\{a_{i_{k-1}},a_{i_k}\}\subseteq
D(R_{i_k},p^{\min}(R))$.
\end{fact}

The formal proof of Fact~\ref{overdemand} is given in \citet{Morimoto15}.

Now we state two lemmas.
The following lemma states that under a mechanism satisfying \textsl{individual rationality} and
\textsl{no subsidy}, the payment of an agent who does not receive any real object is zero.

\begin{lemma}
\label{zero payment}
Let $f$ be a mechanism on $\mathcal{R}^n$ satisfying individual rationality and
no subsidy.
For each $R\in \mathcal{R}^n$ and each $i\in N$, if $a_i(R)=0$, then $t_i(R)=0$.
\end{lemma}

We omit the proof because it is straightforward from individual rationality and no subsidy.
Next,  we introduce a notion of preferences called the ``$(a,t)$-favoring'' preferences.
\begin{defn}
Let $(a,t)\in M\times \mathbb{R}_+$.
A preference relation $R_i\in \mathcal{R}$ is {\bf $(a,t)$-favoring} if for price vector $p$ with $p_{a}=t$ and $p_{b}=0$ for each $b\in L\setminus \{a\}$, $D(R_i,p)=\{a\}$.
\end{defn}

\begin{lemma}
\label{sp and fav}
Let $f$ be a mechanism on $\mathcal{R}^n$ satisfying \textsl{strategy-proofness} and \textsl{no subsidy}.
Let $R\in \mathcal{R}^n$ and $i\in N$ be such that $a_i(R)\ne 0$. Let $R_i'\in \mathcal{R}$ be $f_i(R)$-favoring. Then, $f_i(R_i',R_{-i})=f_i(R)$.
\end{lemma}

For the formal proof of Lemma \ref{sp and fav}, see \citet{Morimoto15}.




\subsection{Proof of Theorem \ref{theo:main}}

Let $f$ be a desirable mechanism on $\mathcal{R}^n$.
Let $R\in \mathcal{R}^n$.
By Fact~\ref{domination}, for each $i\in N$, $t_i(R)\le p^{\min}_{a_i(R)}(R)$.
Hence, the proof of our theorem is completed by establishing that for each $i\in N$, 
$t_i(R)\ge p^{\min}_{a_i(R)}(R)$.
%

Let $\overline{V}\in \mathbb{R}$ be such that
$$\overline{V}>\max_{i\in N}\max_{a\in M}V^{R_i}(a,(0,0)).$$
Note that the right hand side of the inequality is well-defined by finiteness of $N$ and $M$.
Note also that $\overline{V}>0$ by desirability of objects.

Next, we introduce the notion of \emph{individually rational indifference-connected sequence}, which plays an
important role in the proof.

\begin{defn}
Given $(a,t) \in M \times \mathbb{R}$, a pair $S:=(\{i_j\}_{j=1}^k, \{a^j\}_{j=1}^k)$ of sequences of $k\le m$ distinct agents and $k$ distinct real objects is {\bf individually rational indifference-connected (IRIC) sequence from $(a, t)$} if there is a sequence of consumption bundles $\mathbf{z} \equiv \{z^j\}_{j=1}^{k}\equiv \{(a^j, t^j)\}_{j=1}^{k}$ such that 

\begin{enumerate}

\item $z^1 = (a,t)$ i.e., $a^1=a$ and $t^1=t$


\item $z^j~P_{i_j}~(0,0)~\forall~j \in \{1,\ldots,k\}$

\item $z^j~I_{i_j}~z^{j+1}~\forall~j \in \{1,\ldots,k-1\}$.

\end{enumerate}
\end{defn}

\begin{figure}[t!]
\centering
\includegraphics[width=0.8\linewidth]{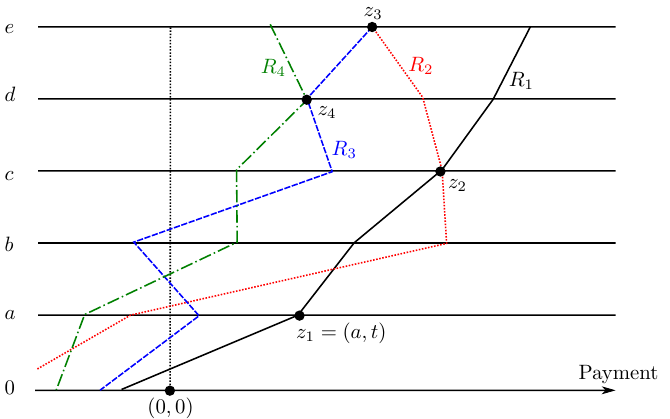}
\caption{Illustration of an IRIC from $(a,t)$.}
\label{fig:IRIC}
\end{figure}

Figure~\ref{fig:IRIC} illustrates an IRIC sequence $S=(\{1,2,3,4\}, \{a, c, e, d\})$ at $(a, t)$. Note that the sequence of consumption bundles $\{z^1,z^2,z^3,z^4\}$ contains distinct agents and distinct real objects.

Let $\mathcal{S}(a,t)$ be the set of all IRIC sequences from $(a,t)$. 
Note that, by continuity, for each $(a, t)\in M\times  \mathbb{R}$ and each $S\in \mathcal{S}(a,t )$, there is $d>t$ such that for each $t'\in \mathbb{R}$ with $t'\le d$, $S$ is also an IRIC from $(a, t')$ for $K$.
Pick such a number for each $z\in M\times\mathbb{R}$ and each $S\in \mathcal{S}(z)$, and denote it by $d(z,S)$.
Given $z\in M\times \mathbb{R}$, let
\begin{equation*}
d(z)\equiv
\begin{cases}
\min\{d(z,S):S \in \mathcal{S}(z)\} & \text{if }\mathcal{S}(z)\neq \emptyset , \\
\overline{V} & \text{otherwise.}
\end{cases}
\end{equation*}
Note that $d(z)$ is well-defined since $\mathcal{S}(z)$ is finite. 
Note also that for each  $z\in M\times \mathbb{R}$, $d(z)\le  \overline{V}$.~\footnote{To see this, let $z\equiv (a, t)\in M\times \mathbb{R}$. If $\mathcal{S}(z)=\emptyset$, then $d(z)=\overline{V}$.
Suppose $\mathcal{S}(z)\neq \emptyset$. 
Let $S=(\{i_j\}_{j=1}^k, \{a^j\}_{j=1}^k)\in \mathcal{S}(z)$.
By the definition of $d(z, S)$, $S$ is also an IRIC sequence from $(a, d(z, S))$.
Thus, by the definition of IRIC sequence, we have $(a, d(z, S))~P_{i_1}~(0, 0)$.
This implies $d(z, S)<V^{R_{i_1}}(a,(0,0))$.
By $V^{R_{i_1}}(a,(0,0))<\overline{V}$ and $d(z)\le  d(z, S)$, we obtain $d(z)\le  \overline{V}$.
}


Now, assume for contradiction that there is an agent $i^*\in N$ such that 
$$t_{i^*}(R)<p^{\min}_{a_{i^*}(R)}(R).$$
By {\sl no subsidy}, $a_{i^*}(R)\ne 0$.

Denote the set of objects as $M\equiv \{x_1,\dots ,x_m\}$ instead of $\{1,\ldots,m\}$ for convenience.
For convenience, we abuse notation and denote $
x_{m+1}\equiv x_1$.
For simplicity of notation, for each $R'\in \mathcal{R}^n$ and each $k\in \{1,\dots
,m+1\}$, we write $p^{\min }_k(R')$ instead of $p^{\min }_{x_k}(R')$.  Without loss of generality, assume $a_{i^*}(R)=x_1$.
Then, using the new notation, we have $t_{i^*}(R)<p^{\min }_1(R)$. 
\begin{step}
\label{pref}
There exists a sequence $\{i_1, \dots, i_m\}$
of $m$ distinct agents 
and a preference profile of these agents $R_{\{i_1,\ldots, i_m\}}'\in \mathcal{R}^m$ such that for the sequence of $(m+1)$ preference profiles
\begin{align*}
    R^0 &= R \\
   R^1 &= (R'_{i_1},R_{-i_1}) \\
    R^2 &= (R'_{\{i_1, i_2\}},R_{-\{i_1, i_2\}}) \\
    \ldots &= \ldots \\
    R^k &= (R'_{\{i_1,\ldots, i_k\}},R_{-\{i_1,\ldots, i_k\}}) \\
    \ldots &= \ldots \\
    R^m &= (R'_{\{i_1,\ldots,i_m\}}, R_{-\{i_1,\ldots,i_m\}})
\end{align*}
the followings hold.
For each $k\in \{1,\dots ,m\}$, $a_{i_k}(R^{k-1})=x_k$, and $R_{i_k}'$ satisfies the
conditions below:\footnote{In the example in Section~\ref{Outline}, the first condition of $R'_1$ corresponds to Condition (1-i) and the second one corresponds to Condition (1-ii).}\\
\begin{enumerate}
\item[($k$-i)] $R_{i_k}'$ is $f_{i_k}(R^{k-1})$-favoring, \\
\item[($k$-ii)] For each $a\in M$,
\begin{align*}
V^{R'_{i_k}}(a,(0,0))\begin{cases}
<d(f_{i_k}(R^{k-1}))&\text{if } a=x_k,\\
>\overline{V}&\text{if }a=x_{k+1},\\
<\overline{V}&\text{if }a\in M\setminus \{x_k, x_{k+1}\},
\end{cases}
\end{align*}
\item[($k$-iii)] for each $a\in M\setminus \{x_k,x_{k+1}\}$,
$V^{R'_{i_k}}(x_{k+1},(a,0))>\overline{V}$.~\footnote{In case of $m=2$, some of the conditions are redundant.
In this case, they vacuously hold, and cause no problem in
the rest of the proof.}
\end{enumerate}
\end{step}
\begin{figure}[t!]
\centering
\includegraphics[width=0.8\linewidth]{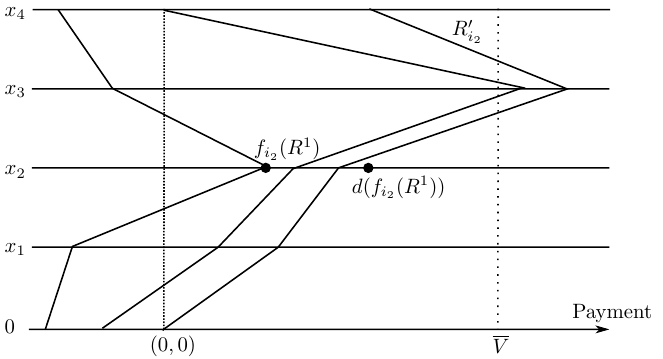}
\caption{Illustration of $R'_{i_2}$.}
\label{fig:Step_1}
\end{figure}

Figure~\ref{fig:Step_1} illustrates $R'_{i_k}$ for $k=2$. \\

\noindent {\bf \sl Proof of Step~\ref{pref}\/}:\enspace 
We inductively construct $\{i_1, \dots, i_m\}$ and $R_{\{i_1, \dots, i_m\}}'$. \\
\\
\textit{Induction base.} Let $i_1=i^*$ (the agent for which we have $t_{i^*}(R) < p_1^{\min}(R)$).
By $R^0=R$, we have $a_{i_1}(R^{0})=a_{i^*}(R)=x_1$.
Note that a preference relation satisfying ($1$-i), ($1$-ii), and ($1$-iii) exists if $t_{i_1}(R)< \overline{V}$.
But this immediately follows because, by individual rationality, $f_{i_1}(R)~R_{i_1}~(0,0)$, and by the definition of $\overline{V}$, $t_{i_1}(R)\le  V^{R_{i_1}}(x_1,(0,0))< \overline{V}$.
\\
\textit{Induction argument.} Let $k\in \{1,\dots, m-1\}$.
Assume that there exist $\{i_1, \dots, i_k\}$ of $k$ distinct agents
and $R_{\{i_1, \dots, i_k\}}'\in \mathcal{R}^k$ such that for each $\ell \in \{1,\dots ,k\}$, $
a_{i_{\ell}}(R^{\ell-1})=a_{\ell}$, and $R'_{i_{\ell}}$ satisfies ($\ell$-i), ($\ell$-ii), and ($\ell$-iii).
By \textsl{no wastage}, there is $i\in N$ such that
\begin{equation}
a_i(R^k)=x_{k+1}. \label{Step 1 x_{j+1}}
\end{equation}
\textsc{Claim:} \textit{$i\notin \{i_1, \dots, i_k\}$.} 
\\
\\
\begin{proof*}
 By contradiction, suppose $i\in \{i_1, \dots, i_k\}$.
By $a_{i_k}(R^{k-1})=x_k$, ($k$-i), and Lemma~\ref{sp and fav}, we have $f_{i_k}(R^k)=f_{i_k}(R^{k-1})$.
Thus, $i\neq i_k$.
Thus, there is $i_{\ell} \in \{i_1, \dots, i_{k-1}\}$ such that $i=i_{\ell}$.
By $\ell\le  k-1<k$, $x_{k+1}\notin \{x_{\ell},x_{\ell+1}\}$.
Thus,
\begin{align}
V^{R'_i}(x_{\ell+1},f_i(R^k))& \ge  V^{R'_i}(x_{\ell+1},(x_{k+1},0))  \tag*{by \eqref{Step 1 x_{j+1}} and \textsl{no subsidy}} \\
& =V^{R'_{i_{\ell}}}(x_{\ell+1},(x_{k+1},0))\hfill   \tag*{by $i=i_{\ell}$} \\
& >\overline{V}.   \tag*{by ($\ell$-iii) and $x_{k+1}\notin  \{x_{\ell},x_{\ell+1}\}$}
\end{align}

By Fact~\ref{domination}, $f_i(R^k)\mathrel{R'_i}(x_{\ell +1},p^{\min }_{\ell+1}(R^k))$.
This implies $V^{R'_i}(x_{\ell+1},f_i(R^k))\le  p^{\min }_{\ell+1}(R^k)$.
Thus, by $V^{R'_i}(x_{\ell+1},f_i(R^k))>\overline{V}$,
\begin{equation}\label{Step 1 more than V}
p^{\min }_{\ell+1}(R^k)>\overline{V}.
\end{equation}

By the definition of $\overline{V}$ and \eqref{Step 1 more than V}, for each $j\in N\setminus \{i_1, \dots, i_k\}$, $V^{R_j}(x_{\ell +1},(0,0))<\overline{V}<p_{\ell+1}^{\min}(R^k)$, which implies  $x_{\ell+1}\notin D(R_j,p^{\min }(R^k))$.
For each $i_{k'} \in \{i_1, \dots, i_k\}\setminus \{i_{\ell},i_{\ell+1}\}$
$$V^{R'_{i_{k'}}}(x_{\ell+1},(0,0))<\overline{V}<p_{\ell+1}^{\min}(R^k),$$
where the first inequality follows from $x_{\ell+1}\notin \{x_{k' },x_{k' +1}\}$ and ($k' $-ii), and the last inequality follows from \eqref{Step 1 more than V}.
Thus, for each $i_{k'}\in \{i_1, \dots, i_k\}\setminus \{i_{\ell},i_{\ell+1}\}$, $x_{\ell+1}\notin D(R'_{i_{k'} },p^{\min}(R^k))$.
Moreover, by ($\ell+1$-ii) and \eqref{Step 1 more than V}, we have $V^{R'_{i_{\ell+1}}}(x_{\ell+1},(0,0))<\overline{V}<p_{\ell+1}^{\min}(R^k)$, implying $x_{\ell+1}\notin D(R'_{i_{\ell+1}},p^{\min }(R^k))$.
Therefore,
\begin{equation*}
|\{j\in N:D(R_j^k,p^{\min }(R^k))\cap
\{x_{\ell+1}\}\neq \emptyset \}|\le  1.
\end{equation*}
This and Fact~\ref{positive} imply that $\{x_{\ell+1}\}$ is weakly underdemanded at $p^{\min}(R^k)$ for $R^k$, contradicting Fact~\ref{overdemand}.
\end{proof*}

Let $i_{k+1}=i$.
By Claim and induction hypothesis, the agents in $\{i_1,\dots, i_k, i_{k+1}\}$ are distinct.
By \eqref{Step 1 x_{j+1}}, $a_{i_{k+1}}(R^k)=x_{k+1}$.
Note that a preference relation satisfying ($k+1$-i), ($k+1$-ii), and ($k+1$-iii) exists if $t_{i_{k+1}}(R^k)<\overline{V}$.
But this immediately follows because, by individual rationality, $f_{i_{k+1}}(R^k)~R_{i_{k+1}}~(0,0)$, and by the definition of $\overline{V}$, $t_{i_{k+1}}(R^k)\le  V^{R_{i_{k+1}}}(x_{k+1},(0,0))< \overline{V}$.
This completes the proof of Step~\ref{pref}.
\hfill $\blacksquare{}$ \vspace{12pt}

Without loss of generality, assume $\{i_1, \dots, i_m\}=\{1, \dots, m\}$.
For each $m'\in \{1,\dots ,m\}$, let $N(m')\equiv \{1,\dots, m'\}$ and $M(m')\equiv \{x_1,\dots ,x_{m'}\}$.
For each $m'\in \{0, 1,\dots ,m\}$, let $p(m')\equiv p^{\min }(R^{m'})$.

\setcounter{equation}{0}
\begin{step}
\label{less than V}
Let $m'\in \{1,\dots ,m\}$ and $x_i\in M$.
Then, $p_i(m')<\overline{V}$.
\end{step}
\begin{figure}[t!]
\centering
\includegraphics[width=0.8\linewidth]{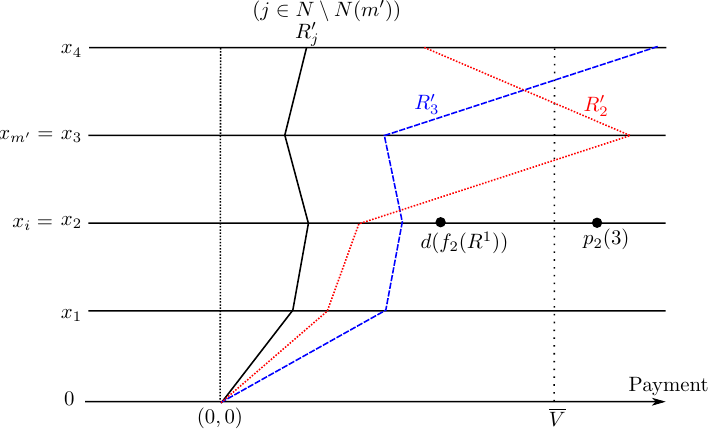}
\caption{Illustration of the proof of Step~\ref{less than V} for $m'=3$ and $i=2$.}
\label{fig:Step_2}
\end{figure}
\noindent {\bf \sl Proof of Step~\ref{less than V}\/}:\enspace 
(See Figure~\ref{fig:Step_2} for illustration.)
 Suppose by contradiction that $p_i(m')\ge
\overline{V}$.
By the definition of $\overline{V}$, for each $j\in N\setminus N(m')$, $V^{R_j}(x_i,(0,0))<\overline{V}\le
p_i(m')$, which implies $(0,0)~P_j~(x_i,p_i(m'))$.
Thus,
\begin{align}
\label{Step 2 j N/N(m')}
x_i\notin D(R_j,p(m'))\text{ for each }j\in N\setminus N(m').
\end{align}

For each $j\in N(m')\setminus \{i-1,i\}$, by $x_i\in M\setminus \{x_j,x_{j+1}\}$ and ($j$-ii) in Step~\ref{pref}, $V^{R'_j}(x_i,(0,0))<\overline{V}\le  p_i(m')$, which implies $(0,0)~P'_j~(x_i,p_i(m'))$.
Thus,
\begin{align}
\label{Step 2 j N(m')/{i-1, i}}
x_i\notin D(R'_j,p(m'))\text{ for each }j\in N(m')\setminus \{i-1,i\}.
\end{align}

Suppose $i\in N(m')$.
By ($i$-ii) in Step~\ref{pref} and the fact that $d(f_i(R^{i-1}))\le  \overline{V}$, we get $V^{R'_i}(x_i,(0,0))< d(f_i(R^{i-1}))\le  \overline{V}\le  p_i(m')$, which implies $(0,0)%
\,P_i'\,(x_i,p_i(m'))$.
Thus,
\begin{align}
\label{Step 2 i}
x_i\notin D(R'_i,p(m')).
\end{align}
Thus, by \eqref{Step 2 j N/N(m')}, \eqref{Step 2 j N(m')/{i-1, i}}, and \eqref{Step 2 i},
\begin{equation*}
|\{j\in N:D(R_j^{m'},p(m'))\cap \{x_i\}\neq
\emptyset \}|\le  1.
\end{equation*}%
If $i\notin N(m')$, this inequality is immediately implied by \eqref{Step 2 j N/N(m')} and \eqref{Step 2 j N(m')/{i-1, i}}.

Hence, in either case, $\{x_i\}$ is a weakly underdemanded set at $p(m')$ for $
R^{m'}$. This contradicts Fact~\ref{overdemand}.
This completes the proof of Step~\ref{less than V}.
\hfill $\blacksquare{}$ \vspace{12pt}
\setcounter{equation}{0}
\begin{step}
\label{N(m')'s demand}
Let $m'\in \{1,\dots, m\}$, $((a_i,t_i))_{i\in N}\in Z^{\min}(R^{m'})$, and $i\in N(m')$.
Then, the following properties hold.
\begin{enumerate}
\renewcommand{\theenumi}{\roman{enumi}}
\renewcommand{\labelenumi}{(\theenumi)}
\item (Twin demand property.)
$D(R'_i,p(m'))\subseteq \{x_i,x_{i+1}\}$, and thus, $a_i\in \{x_i, x_{i+1}\}$.
\item (Unique demand property.) If $x_i\notin M_+(m')$, $D(R'_i,p(m'))=\{x_i\}$, and thus, $a_i = x_i$.
\end{enumerate}
\end{step}
\noindent {\bf \sl Proof of Step~\ref{N(m')'s demand}\/}:\enspace 
We show the first property.
By Step~\ref{less than V} and ($i$-ii) in Step~\ref{pref}, $p_{i+1}(m')<\overline{V}<V^{R'_i}(x_{i+1},(0,0))$, which implies $(x_{i+1},p_{i+1}(m'))\,P_i'\,(0,0)$.
Thus, $0\notin D(R'_i,p(m'))$.

Let $a\in M\setminus \{x_i,x_{i+1}\}$. By Step~\ref{less than V}, ($i$-iii) in Step~\ref{pref}, and $
p_a(m')\ge  0$, we have
\begin{equation*}
p_{i+1}(m')<\overline{V}<V^{R'_i}(x_{i+1},(a,0))\le
V^{R'_i}(x_{i+1},(a,p_a(m'))),
\end{equation*}%
which implies $(x_{i+1},p_{i+1}(m'))\,P_i'\,(a,p_a(m'))$. Thus, $a\notin D(R'_i,p(m'))$
. Hence, $D(R'_i,p(m'))\subseteq \{x_i,x_{i+1}\}$, and this immediately implies $a_i\in \{x_i, x_{i+1}\}$.

Next, we show the second property. Suppose $x_i\notin M_+(m')$.
Then, $p_i(m')\le t_i(R^{i-1})$.
Since $R'_i$ is $f_i(R^{i-1})$-favoring, for each $a\in L\setminus \{x_i\}$, 
$$
V^{R'_i}(a,(x_i, p_i(m')))\le  V^{R'_i}(a, f_i(R^{i-1}))<0\le  p_a(m'),
$$ 
implying that $(x_i,p_i(m'))~P'_i~(a,p_a(m'))$.
Therefore, for each $a\in L\setminus \{x_i\}$, $a\notin D(R'_i,p(m'))$.
Hence, $D(R'_i,p(m'))=\{x_i\}$, and this immediately implies $a_i=x_i$.
\hfill $\blacksquare{}$ \vspace{12pt} 



Now we introduce two notations.
Given $m'\in \{1,\dots ,m\}$, let
\begin{align*}
& M_+(m')\equiv \{x_k\in M(m'):p_k(m')>t_k(R^{k-1})\}\text{, and} \\
& M_{++}(m')\equiv \{x_k\in M(m'):p_k(m')>V^{R'_k}(x_k,(0,0))\}.
\end{align*}
\begin{figure}[!hbt]
\centering
\includegraphics[width=0.8\linewidth]{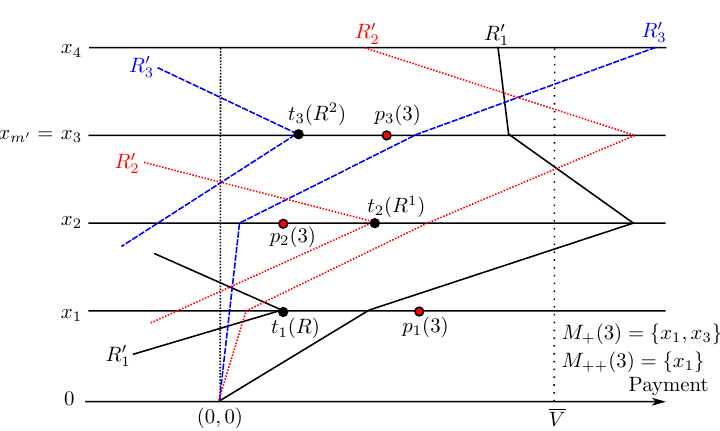}
\caption{$M_+(3)$ and $M_{++}(3)$.}
\label{fig:M_+(m')}
\end{figure}

Figure~\ref{fig:M_+(m')} illustrates $M_+(m')$ and $M_{++}(m')$ for $m'=3$.
In this figure, $p_1(3)>V'_1(x_1,(0,0))$, $p_2(3)<t_2(R^1)$, and $t_3(R^2)<p_3(3)<V'_3(x_3,(0,0))$.
Thus, $M_+(3)=\{x_1,x_3\}$ and $M_{++}(3)=\{x_1\}$.

Note that for each $m'\in \{1,\dots ,m\}$, $M_{++}(m')\subseteq M_+(m')$.
To see this, let $m'\in \{1,\dots ,m\}$ and $x_k\in M_{++}(m')$.
Then, $p_k(m')>V^{R'_k}(x_k,(0,0))$.
By ($k$-i) in Step~\ref{pref}, $R'_k$ is $f_k(R^{k-1})$-favoring.
Thus, $t_k(R^{k-1})<V^{R'_k}(x_k, (0,0))<p_k(m')$.
Thus, $x  _k\in M_+(m')$.



\setcounter{equation}{0}
\begin{step}[Outside demander]
\label{i demands M_+}
Let $m'\in \{1,\dots ,m\}$ be such that $
M_+(m')\neq \emptyset $. Then, there is $i\in N\setminus
N(m')$ such that $D(R_i,p(m'))\cap M_+(m')\neq \emptyset $.
\end{step}
\noindent {\bf \sl Proof of Step~\ref{i demands M_+}\/}:\enspace 
Suppose for contradiction that 
\begin{equation}
\label{Step 4 in Nsetminus N(m')}
D(R_i,p(m'))\cap M_+(m')=\emptyset \text{ for each } i\in N\setminus N(m').
\end{equation}
By Step~\ref{N(m')'s demand}~(ii) (Unique demand property),
\begin{equation}
\label{Step 4 N(m')}
\{i\in N(m'):x_i\notin M_+(m')\}\subseteq \{i\in
N(m'):D(R'_i,p(m'))\cap M_+(m')=\emptyset \}.
\end{equation}
By \eqref{Step 4 in Nsetminus N(m')} and \eqref{Step 4 N(m')},
\begin{align}
 |\{i\in N:&D(R_i^{m'},p(m'))\cap M_+(m')\neq \emptyset \}|  \notag \\
=& |\{i\in N(m'):D(R'_i,p(m'))\cap
M_+(m')\neq \emptyset \}|  \tag*{by \eqref{Step 4 in Nsetminus N(m')}} \\
=& |N(m')|-|\{i\in N(m'):D(R'_i,p(m'))\cap M_+(m')=\emptyset \}|  \notag \\
\le  & |N(m')|-|\{i\in N(m'):x_i\notin M_+(m')\}|  \tag*{by \eqref{Step 4 N(m')}} \\
=& |M(m')|-|M(m')\setminus M_+(m')| \notag \\ 
=& |M_+(m')|.
\tag*{by
$M_+(m') \subseteq M(m')$}
\end{align}
Thus, this inequality and Fact~\ref{positive} imply that $M_+(m')$ is weakly underdemanded at $p(m')$ for $
R^{m'}$. This contradicts Fact~\ref{overdemand}.
\hfill $\blacksquare{}$ \vspace{12pt} 

\setcounter{equation}{0}
\begin{step}
\label{property of IRIC sequence}
Let $m'\in \{1,\dots ,m\}$.
Let $\{i_1,\dots ,i_K\}\subseteq N\setminus N(m')$ and $
\{b_1,\dots ,b_K\}\subseteq M$ be such that\\
(a) $b_1\in M_+(m')\setminus M_{++}(m')$,\\
(b) for each $k\in \{1,\dots ,K-1\}$, $\{b_k,b_{k+1}\}\subseteq
D(R_{i_k},p(m'))$, and\\
(c) $b_K\in D(R_{i_K},p(m'))$.\\
Then, $0\notin D(R_{i_K},p(m'))$.
\end{step}
\noindent {\bf \sl Proof of Step~\ref{property of IRIC sequence}\/}:\enspace 
 By (b), for each $k\in \{1,\dots ,K-1\}$, $(b_k,p_{b_k}(m'))\,I_{i_k}\,(b_{k+1},p_{b_{k+1}}(m'))$.
 Thus,
 \begin{align}
 \label{Step 5 p^{b_k}=V_{i_{k-1}}}
 p_{b_{k+1}}(m')=V^{R_{i_k}}(b_{k+1},(b_k,p_{b_k}(m')))\text{ for each }k\in \{1,\dots ,K-1\}.
 \end{align}

 We first show that $(\{i_1,\dots, i_K\}, \{b_1,\dots, b_K\})$ is an IRIC from $f_j(R^{j-1})$. 
 By $b_1\in M_+(m')$, there is $j\in \{1,\dots
,m'\}$ such that $b_1=x_j$. 
Let $r_1\equiv t_j(R^{j-1})$,
and for each $k\in \{2,\dots ,K\}$, let $r_k\equiv
V^{R_{i_{k-1}}}(b_k,(b_{k-1},r_{k-1}))$.
Then, $(b_1, r_1)=f_j(R^{j-1})$, and for each $k\in \{1,\dots, K-1\}$, $(b_k, r_k)~I_{i_k}(b_{k+1}, r_{k+1})$.
\\
\\
\textsc{Claim:} \textit{For each $k\in \{1,\dots ,K\}$, $
r_k<p_{b_k}(m')$.} \\
\\
\begin{proof*} The proof is by induction.\\
\textit{Induction base.} Let $k=1$. By (a) and $b_1=x_j$, $
r_1=t_j(R^{j-1})<p_j(m')$. \\
\textit{Induction argument.} Let $k\ge 1$ and assume that $r_k<p_{b_k}(m')$.
Then,
\begin{align}
p_{b_{k+1}}(m')& =V^{R_{i_k}}(b_{k+1},(b_k,p_{b_k}(m')))  \tag*{ by \eqref{Step 5 p^{b_k}=V_{i_{k-1}}}} \\
& >V^{R_{i_k}}(b_{k+1},(b_k,r_k))
\tag*{by
$r_k<p_{b_k}(m' )$} \\
& =r_{k+1}  \notag
\end{align}
\end{proof*}
 
By (b) and (c), for each $k\in \{1,\dots ,K\}$, $(b_k,p_{b_k}(m'))\,R_{i_k}\,(0,0)$. Thus, by Claim, for each $k\in \{1,\dots
,K\}$,
$(b_k,r_k)\,P_{i_k}\,(b_k,p_{b_k}(m'))\,R_{i_k}\,(0,0)$.
Therefore, $(\{i_1,\dots, i_K\}, \{b_1,\dots, b_K\})$ is an IRIC from $f_j(R^{j-1})$.

By (a) and ($j$-ii) in Step~\ref{pref},
\begin{equation*}
\label{Step 5 less than V'_j}
p_{b_1}(m')\le  V^{R'_j}(x_j,(0,0))<d(f_j(R^{j-1})).
\end{equation*}
Therefore, by the definition of $d(f_j(R^{j-1}))$, $(\{i_1,\dots, i_K\}, \{b_1,\dots, b_K\})$ is an IRIC from $(x_j, p_j(m'))$.
By \eqref{Step 5 p^{b_k}=V_{i_{k-1}}}, the corresponding sequence of bundles is $\{(b_k, p_{b_k}(m')\}_{k=1}^K$.
Therefore, 
$$
(b_K, p_{b_K}(m'))~P_{i_K}(0, 0).
$$
Hence, $0\notin D(R_{i_K}, p(m'))$, and this completes the proof of Step~\ref{property of IRIC sequence}.

\hfill $\blacksquare{}$ \vspace{12pt} 
\setcounter{equation}{0}
\begin{step}[Outside receiver I]
\label{x_ in M_+(m')}
Let $m'\in \{1,\dots ,m-1\}$ be such that $
M_+(m')\neq \emptyset $ and $M_{++}(m')=\emptyset $. Let
$((a_i,t_i))_{i\in N}\in Z^{\min}(R^{m'})$. Then, there exists $i\in N\setminus N(m')$ such that $a_i\in M_+(m')$.
\end{step}
\noindent {\bf \sl Proof of Step~\ref{x_ in M_+(m')}\/}:\enspace 
Suppose for contradiction that for each $i\in N\setminus N(m')$, $a_i\notin M_+(m')$.
By Step~\ref{N(m')'s demand}~(ii) (Unique demand property), for each $x_j\in M(m')\setminus M_+(m')$, $x_j=a_j$.
Thus, 
\begin{equation}
\label{Step 6 a_i notin M(m')}
\{i\in N\setminus N(m'):a_i\in M(m')\}=\emptyset .
\end{equation}
By Step~\ref{N(m')'s demand}~(i) (Twin demand property), $m'<m$ and \eqref{Step 6 a_i notin M(m')},~\footnote{The proof is as follows. By no wastage, there is $i\in N$ such that $a_i=x_1$. By \eqref{Step 6 a_i notin M(m')}, $i\in N(m')$. By Step~\ref{N(m')'s demand}~(i) (Twin demand property) and $m'<m$, $i\notin \{2,\dots, m'\}$.
Thus, $i=1$.
Then, by Step~\ref{N(m')'s demand}~(i) (Twin demand property), we can inductively show that $a_2=x_2$, $a_3=x_3$,..., $a_{m'}=x_{m'}$.}
\begin{align}
\label{Step 6 a_i=x_i}
a_i=x_i\text{ for each }i\in N(m').
\end{align}

By Step~\ref{i demands M_+} (Outside demander) and $M_+(m')\neq \emptyset $, there exist $i\in
N\setminus N(m')$ and $x_j\in M_+(m')$ such that $
x_j\in D(R_i,p(m'))$. By Fact~\ref{path}, there is a sequence $
\{i_k\}_{k=1}^K$ of $K$ distinct agents such that\footnote{Note that it is possible that $a_i=0$.
In this case, the sequence $\{i_k\}_{k=1}^K$ consists only of agent $i$. That is, $K=1$ and $\{i_k\}_{k=1}^K=\{i\}$.}
\begin{align}
\label{Step 6 i_1}
& i_1=i,  \\
\label{Step 6 i_K}& a_{i_K}=0\text{ and } \text{for each }k\in \{1,\dots, K-1\},\ a_{i_k}\ne 0\text{ and}   \\
\label{Step 6 dc}
& \text{for each }k\in \{2,\dots ,K\},\{a_{i_{k-1}},a_{i_k}\}\subseteq D(R_{i_k}^{m'},p(m')).
\end{align}

\begin{figure}[t!]
\centering
\includegraphics[width=0.8\linewidth]{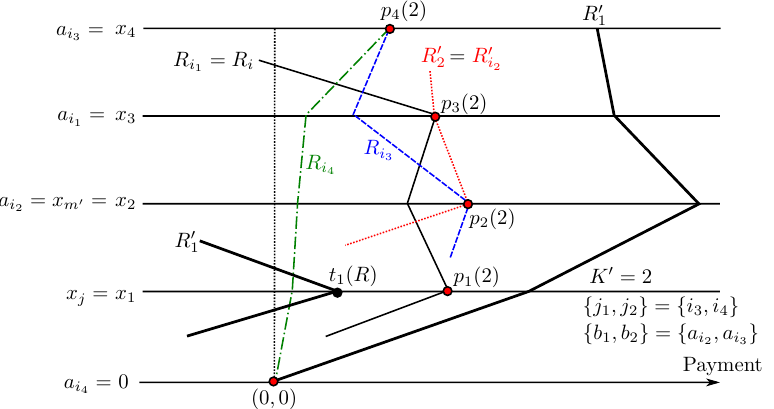}
\caption{$\{i_k\}_{k=1}^K$ in the proof of Step 6.}
\label{fig:Step6}
\end{figure}

Figure~\ref{fig:Step6} illustrates the sequence $
\{i_k\}_{k=1}^K$, for $m'=2$, $K=4$, and $x_j=x_1$.

Let $i_0\equiv j$.
By $x_j\in M(m')$, 
$i_0\in N(m')$.
By \eqref{Step 6 a_i=x_i} and \eqref{Step 6 i_K}, $i_K\notin N(m')$.
Thus, there is $K'\in \{0,1,\dots ,K-1\}$ such that
\begin{align}
\label{Step 6 i_{K'} in N(m')}
& i_{K'}\in N(m'),\text{ and}   \\
& \{i_{K'+1},\dots ,i_K\}\cap N(m')=\emptyset .  \label{Step 6 disjoint}
\end{align}

Let $\{j_1,\dots ,j_{K-K'}\}\subseteq N$ and $\{b_1,\dots ,b_{K-K'}\}\subseteq M$ be such that for each $k\in \{1,\dots ,K-K'\}$,
\begin{equation*}
j_k=i_{K'+k}\text{ and }b_k=a_{i_{K'+k-1}}.
\end{equation*}
In Figure~\ref{fig:Step6}, $K'=2$ and $K-K'=2$.
Thus, $\{j_1,j_2\}=\{i_3,i_4\}$ and $\{b_1,b_2\}=\{a_{i_2}, a_{i_3}\}=\{x_2, x_4\}$.
\\
\\
\textsc{Claim:}
\textit{$\{j_1,\dots ,j_{K-K'}\}\subseteq N\setminus N(m')$, $\{b_1, \dots, b_{K-K'}\}\subseteq M$, and the pair $(\{j_k\}_{k=1}^{K-K'},\{b_k\}_{k=1}^{K-K'})$
satisfies the following conditions.
\\
(a) $b_1\in M_+(m')\setminus M_{++}(m')$,\\
(b) for each $k\in \{1,\dots ,K-K'-1\}$, $\{b_k,b_{k+1}\}\subseteq D(R_{j_k},p(m'))$, and\\
(c) $b_{K-K'}\in D(R_{j_{K-K'}},p(m'))$.
}
\\
\\
\begin{proof*}
By $\{j_1,\dots ,j_{K-K'}\}=\{i_{K'+1},\dots, i_K\}$ and \eqref{Step 6 disjoint}, $\{j_1,\dots ,j_{K-K'}\}\subseteq N\setminus N(m')$.
By Step~\ref{N(m')'s demand}~(i) (Twin demand property) and $i_0=j\in N(m'),$ $a_{i_0}\ne 0$.
By $\{b_1,\dots, b_{K-K'}\}=\{a_{i_{K'}},\dots, a_{i_{K-1}}\}$, $a_{i_0}\ne 0$, and \eqref{Step 6 dc}, $\{b_1,\dots, b_{K-K'}\}\subseteq M$.
\\
\textsl{Proof of (a):}
Note that by $M_{++}(m')=\emptyset$, (a) is equivalent to $b_1\in M_+(m')$.
If $K'=0$, then by \eqref{Step 6 a_i=x_i}, $b_1=a_{i_0}=a_j=x_j\in M_+(m')$.

Suppose $K'\ge  1$.
By $i_1=i\notin N(m')$ and \eqref{Step 6 i_{K'} in N(m')}, $K'>1$.
Then, by \eqref{Step 6 dc},
$$|D(R'_{i_{K'}},p(m'))|\ge  2.$$
By \eqref{Step 6 i_{K'} in N(m')} and \eqref{Step 6 a_i=x_i},
$$b_1=a_{i_{K'}}=x_{i_{K'}}.$$
Thus, Step~\ref{N(m')'s demand}~(ii) (Unique demand property) implies $b_1\in M_+(m')$.
\\
\textsl{Proof of (b):}
Let $k\in \{1,\dots ,K-K'-1\}$.
First, suppose $K'=0$ and $k=1$.
Note that $\{b_1,b_2\}=\{a_{i_0}, a_{i_1}\}$.
By $i_0=j$, $j\in N(m')$, and \eqref{Step 6 a_i=x_i}, $a_{i_0}=x_j$.
Thus, by \eqref{Step 6 dc},
$$\{b_1,b_2\}=\{x_j,a_i\}\subseteq  D(R_i,p(m'))=D(R_{j_1},p(m')).$$

Next, suppose either $K'\neq 0$ or $k\neq 1$.
In both cases, we have $K'+k>1$.
Thus, $j_k=i_{K'+k}\neq i_1$.
Therefore, by \eqref{Step 6 dc},
$$\{b_k, b_{k+1}\}=\{a_{i_{K'+k-1}},a_{i_{K'+k}}\}\subseteq D(R_{i_{K'+k}},p(m'))=D(R_{j_k},p(m')).$$
\textsl{Proof of (c):}
By $j_{K-K'}=i_K$, $b_{K-K'}=a_{i_{K-1}}$, and \eqref{Step 6 dc},
$$b_{K-K'}=a_{i_{K-1}}\in D(R_{i_K},p(m'))=D(R_{j_{K-K'}},p(m')).$$
\end{proof*}

By Step~\ref{property of IRIC sequence} and Claim, $0\notin D(R_{j_{K-K'}},p(m'))=D(R_{i_K},p(m'))$.
This contradicts \eqref{Step 6 i_K}.
This completes the proof of Step~\ref{x_ in M_+(m')}.
\hfill $\blacksquare{}$ \vspace{12pt} 
\setcounter{equation}{0}
\begin{step}[Outside receiver II]
\label{x_i in M_+(k)}
Let $m'\in \{1,\dots ,m-1\}$ be such that $M_{++}(m')\neq \emptyset$.
Let $x_k\in M_{++}(m')$ and $((a_i,t_i))_{i\in N}\in Z^{\min}(R^{m'})$.
Then, there exists $i\in N\setminus N(m')$ such that $a_i\in M(k)\cap M_+(m')$.
\end{step}
\noindent {\bf \sl Proof of Step~\ref{x_i in M_+(k)}\/}:\enspace 
 Suppose for contradiction that for each $i\in N\setminus N(m')$, $a_i\notin M(k)\cap M_+(m')$. 
By Step~\ref{N(m')'s demand}~(ii) (Unique demand property), for each $j\in N(k)$ with $x_j\notin M_+(m')$, $a_j=x_j$.
 Thus,
\begin{equation}
\label{x_i not in M(k)}
\{i\in N\setminus N(m'):a_i\in M(k)\}=\emptyset .
\end{equation}

Note that by $k\ge  1$, $x_1\in M(k)$. Thus by \eqref{x_i not in M(k)}, for each $i\in N\setminus N(m')$, $a_i\neq x_1$.
By $m'<  m$ and Step~\ref{N(m')'s demand}~(i) (Twin demand property), for each $i\in N(m')\setminus \{1\}$, $a_i\neq x_1$.
Thus, by \textsl{no wastage}, we conclude that $a_1=x_1.$
By using \eqref{x_i not in M(k)} and Step~\ref{N(m')'s demand}~(i) (Twin demand property) repeatedly, we obtain
$$a_2=x_2,\ a_3=x_3,\ \dots,\ a_k=x_k.$$

By $x_k\in M_{++}(m')$, $p_k(m')>V^{R'_k}(x_k,(0,0))$. Thus,
\begin{equation*}
(0,0)\,I_k'\,(x_k,V^{R'_k}(x_k,(0,0)))~P_k'~(x_k,p_k(m')).
\end{equation*}
This implies $x_k\notin D(R_k',p(m'))$, contradicting $a_k=x_k$. 
\hfill $\blacksquare{}$ \vspace{12pt} 
\setcounter{equation}{0}
\begin{step}
\label{p^{m'+1}>t_{m'+1}}
Let $m'\in \{0,\dots ,m-1\}$ be such that $M_+(m')\neq \emptyset $.
Let $((a_i,t_i))_{i\in N}\in Z^{\min}(R^{m'})$.
Let $\{i_k\}_{k=1}^K$ be a sequence of $K$ distinct agents such that
\begin{align}
& i_1=m'+1,   \label{Step 8 i_1}\\
& a_{i_K}=0\text{ and } \text{for each }k\in \{1,\dots, K-1\},\ a_{i_k}\ne 0\text{ and} \label{Step 8 i_K} \\
& \text{for each }k\in \{2,\dots ,K\},\{a_{i_{k-1}},a_{i_k}\}\subseteq
D(R_{i_k}^{m'},p(m')). \label{Step 8 dc}
\end{align}%
Suppose $m'\in \{i_1,\dots ,i_K\}$.\footnote{Note that this implies $K>1$.} Then $p_{m'+1}(m')>t_{m'+1}(R^{m'})$.
\end{step}
\noindent {\bf \sl Proof of Step~\ref{p^{m'+1}>t_{m'+1}}\/}:\enspace 
This step will be proved using five claims, which we state and prove as we go along the proof of this step. Assume for contradiction that $p_{m'+1}(m')\le  t_{m'+1}(R^{m'})$.
By $M_+(m')\neq \emptyset $, and Steps \ref{x_ in M_+(m')} (Outside receiver I) and \ref{x_i in M_+(k)} (Outside receiver II), there exist $i\in N\setminus N(m')$ and $x_{\ell}\in M_+(m')$ such that~\footnote{If $M_{++}(m')=\emptyset$, then Step \ref{x_ in M_+(m')} (Outside receiver I) implies there is $i\in N\setminus N(m')$ such that $x_i\in M_+(m')$. If there is $x_k\in M$ such that $x_k\in M_{++}(m')$, then Step \ref{x_i in M_+(k)} (Outside receiver II) implies that there is $i\in N\setminus N(m')$ such that $x_i\in M(k)\cap M_+(m')\subseteq M_+(m')$.}
\begin{equation}
\label{Step 8 x_ in M_+}
a_i=x_{\ell}.
\end{equation}
Then, by using Step~\ref{N(m')'s demand}~(i) (Twin demand property) repeatedly, we obtain
\begin{align}
\label{Step 8 a_k=x_{k+1}}
a_{\ell}=x_{\ell +1},\ a_{\ell +1}=x_{\ell +2},\ \dots,\ a_{m'}=x_{m'+1}.
\end{align}
\textsc{Claim 1:}
\textit{Let $k\in \{1,\dots ,K-1\}$ be such that $i_k\notin N(m')$ and $i_{k+1}\in N(m')$. Then $i_k=i$.}
\\
\\
\begin{proof*}
By $i_{k+1}\in N(m')$, \eqref{Step 8 dc}, and Step~\ref{N(m')'s demand}~(i) (Twin demand property), $a_{i_k}\in D(R'_{i_{k+1}},p(m'))\subseteq M(m'+1)$.
By Step~\ref{N(m')'s demand}~(i) (Twin demand property) and $a_i\in M_+(m')$,
$\{a_j:j\in N(m')\cup \{i\}\}=M(m'+1)$.
Thus, by $i_k\notin N(m')$, $i_k=i$.
\end{proof*} \\

By $m'\in \{i_1,\dots ,i_K\}$, $\{i_1,\dots ,i_K\}\cap N(m')\neq \emptyset $.
Thus, by \eqref{Step 8 i_1}, there is $K'\in \{1,\dots ,K-1\}$ such that
\begin{equation*}
\{i_1,\dots ,i_{K'}\}\cap N(m')=\emptyset \text{ and }%
i_{K'+1}\in N(m').
\end{equation*}%
By Claim 1,
\begin{equation}
\label{Step 8 i_{K'}=i}
i_{K'}=i.
\end{equation}
By $m'\in \{i_1,\dots ,i_K\}$, there is $K''\in \{1,\dots ,K\}$ such that
$$i_{K''}=m'.$$
By $\{i_1,\dots ,i_{K'}\}\cap N(m')=\emptyset $ and $i_{K''}=m'\in N(m')$, we have $ K''>K'$.
Note that by $a_{i_{K''}}=x_{m'+1}$ and \eqref{Step 8 i_K}, $K''<K$.
Therefore,
\begin{equation*}
K'<K''<K.
\end{equation*}
\begin{figure}[!hbt]
\centering
\includegraphics[width=0.9\linewidth]{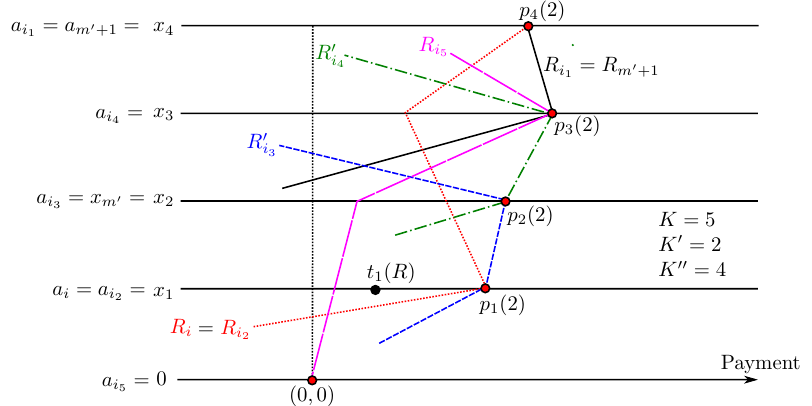}
\caption{Illustration of $\{i_k\}_{k=1}^K$ for $m'=2$ and $K=5$.}
\label{fig:Step8_1}
\end{figure}
Figure~\ref{fig:Step8_1} is an illustration of the sequence $\{i_k\}_{k=1}^K$ for $m'=2$ and $K=5$.
In this figure, $K'=2$ and $K''=4$.
\\
\\
\textsc{Claim 2:} \textit{$\{i_{K''+1},\dots ,i_K\}\cap
N(m')=\emptyset $.} \\
\\
\begin{proof*}
First, we show $i_{K''+1}\notin N(m')$.
Suppose for contradiction $i_{K''+1}\in N(m')$.
By $i_{K''}=m'$, $i_{K^{\prime\prime }+1}\in N(m'-1)$.
Thus, by Step~\ref{N(m')'s demand}~(i) (Twin demand property), $x_{m'+1}\notin D(R'_{i_{K''+1}},p(m'))$.
However, by \eqref{Step 8 a_k=x_{k+1}} and \eqref{Step 8 dc},
\begin{equation*}
x_{m'+1}=a_{i_{K''}}\in D(R'_{i_{K''+1}},p(m')),
\end{equation*}
a contradiction.
Hence, $i_{K''+1}\notin N(m')$.

Now, suppose for contradiction that $\{i_{K''+1},\dots ,i_K\}\cap N(m')\neq \emptyset $.
Then, by $i_{K''+1}\notin N(m')$, there is $K^*\in \{K''+1,\dots ,K-1\}$ such that $i_{K^*}\notin N(m')$ and $i_{K^*+1}\in N(m')$.
By Claim 1, we have $i_{K^*}=i$.
By \eqref{Step 8 i_{K'}=i}, $i_{K^*}=i_{K'}$.
Since agents $\{i_1,\dots, i_K\}$ are distinct, $K^*=K'$.
However, this is a contradiction because $K'<K''<K^*$.
\end{proof*}
\\
\textsc{Claim 3:}
\textit{$\{\ell, \dots, m'\}\subseteq \{i_{K'+1},\dots, i_{K''}\}$.}
\\
\\
\begin{proof*}
Suppose for contradiction that there is $k\in \{\ell, \dots, m'\}$ such that $k\notin \{i_{K'+1},\dots, i_{K''}\}$.
Without loss of generality, we can let $k$ be the largest number in $\{\ell, \dots, m'\}$ such that $k\notin \{i_{K'+1},\dots, i_{K''}\}$.
By $i_{K''}=m'$, $k<m'$.
Thus, there is $K^*\in \{K'+1,\dots, K''\}$ such that $i_{K^*}=k+1$.

By \eqref{Step 8 dc} and Step~\ref{N(m')'s demand}~(i) (Twin demand property),
$$\{a_{i_{K^*-1}},a_{i_{K^*}}\}\subseteq D(R^{m'}_{i_{K^*}},p(m'))= D(R'_{k+1},p(m'))\subseteq \{x_{k+1},x_{k+2}\}.$$
Further, by \eqref{Step 8 a_k=x_{k+1}}, $a_{i_{K^*}}=x_{k+2}$.
Thus, $a_{i_{K^*-1}}=x_{k+1}$.
Hence, by \eqref{Step  8 a_k=x_{k+1}}, $i_{K^*-1}=k$.

By $k\notin \{i_{K'+1},\dots, i_{K''}\}$ and $K^*\in \{K'+1,\dots, K''\}$, $K^*-1=K'$.
This and \eqref{Step 8  i_{K'}=i} imply that $i=k$.
However, this is a contradiction since $i\notin N(m')$ and $k\in N(m')$.
\end{proof*}
\\
\begin{figure}[!hbt]
\centering
\includegraphics[width=0.7\linewidth]{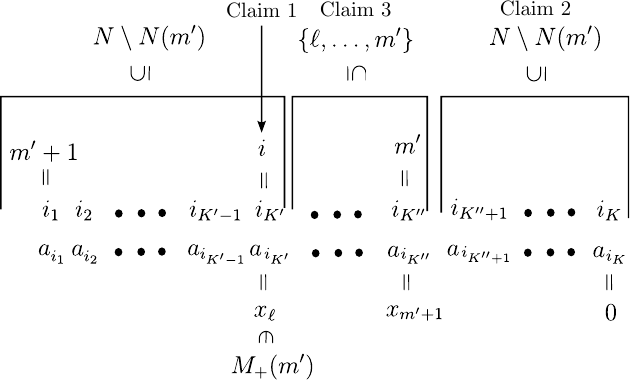}
\caption{Properties of $\{i_k\}_{k=1}^K$ and $\{a_k\}_{k=1}^K$.}
\label{fig:Step8_2}
\end{figure}
Figure~\ref{fig:Step8_2} summarizes the properties of the sequences $\{i_k\}_{k=1}^K$ and $\{a_k\}_{k=1}^K$, which we have uncovered so far in the proof.
\\
\\
\textsc{Claim 4:}
\textit{$M_{++}(m')=\emptyset$.}
\\
\\
\begin{proof*}
Suppose for contradiction that $M_{++}(m')\neq \emptyset$.
Let $ x_{k^*}\in M_{++}(m')$.
Note that Step~\ref{N(m')'s demand}~(i) (Twin demand property), there is at most one agent $j\in N\setminus N(m')$ such that $a_j\in M(m')$.
Thus, by $a_i\in M(m')$, $i\in N\setminus N(m')$, and Step~\ref{x_i in M_+(k)} (Outside receiver II), $a_i\in M(k^*)\cap M_+(m')$.
Therefore, by $a_i=x_{\ell}$, $\ell\le  k^*$.

By Claim 3, there is $K^*\in \{K'+1,\dots, K''\}$ such that $i_{K^*}=k^*$.
By \eqref{Step 8 dc} and Step~\ref{N(m')'s demand}~(i) (Twin demand property),
$$\{a_{i_{K^*-1}},a_{i_{K^*}}\}\subseteq D(R^{m'}_{i_{K^*}},p(m'))= D(R'_{k^*},p(m'))\subseteq \{x_{k^*},x_{k^*+1}\},$$
which implies $\{a_{i_{K^*-1}},a_{i_{K^*}}\}=\{x_{k^*},x_{k^*+1}\}$.
Thus, $x_{k^*}\in D(R'_{k^*},p(m'))$.
However, by $x_{k^*}\in M_{++}(m')$, $p_{k^*}(m')>V^{R'_k}(x_{k^*},(0,0))$.
This implies $(0,0)~P'_{k^*}~(x_{k^*},p(m'))$, and thus, $x_{k^*}\notin D(R'_{k^*},p(m'))$, a contradiction.
\end{proof*}

We now construct two new sequences $\{j_k\}_{k=1}^{K'+K-K''}$ and $\{b_k\}_{k=1}^{K'+K-K''}$ using 
the sequences $\{i_k\}_{k=1}^K$ and $\{a_k\}_{k=1}^K$.
\begin{figure}[!hbt]
\centering
\includegraphics[width=0.9\linewidth]{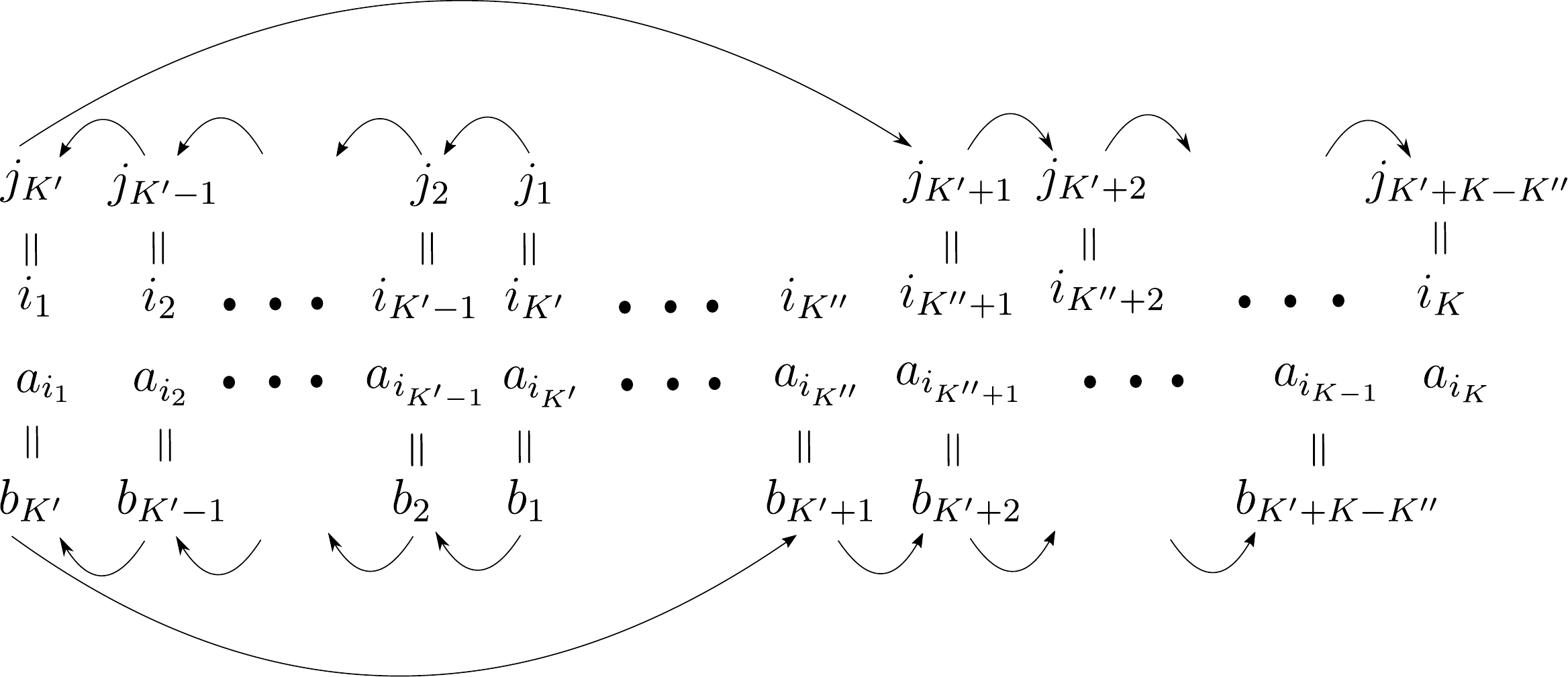}
\caption{Illustration of $\{j_k\}_{k=1}^{K'+K-K''}$ and $\{b_k\}_{k=1}^{K'+K-K''}$.}
\label{fig:Step8_3}
\end{figure}
Let $\{j_k\}_{k=1}^{K'+K-K''}$ and $\{b_k\}_{k=1}^{K'+K-K''}$ be such that for each $k\in \{1,\dots ,K'+K-K''\}$,
\begin{equation*}
j_k=
\begin{cases}
i_{K'+1-k} & \text{if }k\le  K', \\
i_{K''+k-K'} & \text{if }k>K',
\end{cases}
\text{ and }b_k=
\begin{cases}
a_{i_{K'+1-k}} & \text{if }k\le  K', \\
a_{i_{K''+k-K'-1}} & \text{if }k> K'.
\end{cases}%
\end{equation*}

Figure \ref{fig:Step8_3} is an illustration of $\{j_k\}_{k=1}^{K'+K-K''}$ and $\{b_k\}_{k=1}^{K'+K-K''}$.
\\
\\
\textsc{Claim 5:}
\textit{$\{j_1,\dots ,j_{K'+K-K"}\}\subseteq N\setminus N(m')$ and $(\{j_k\}_{k=1}^{K'+K-K''}, \{b_k\}_{k=1}^{K'+K-K''})$ satisfies the following conditions.\\
(a) $b_1\in M_+(m')\setminus M_{++}(m')$,\\
(b) for each $k\in \{1,\dots ,K'+K-K''-1\}$, $\{b_k,b_{k+1}\}\subseteq
D(R_{j_k},p(m'))$, and\\
(c) $b_{K'+K-K''}\in D(R_{i_{K'+K-K''}},p(m'))$. }
\\
\\
\begin{proof*}
First, we show $\{j_1,\dots ,j_{K'+K-K"}\}\subseteq N\setminus N(m')$.
Let $k\in \{1,\dots ,K'\}$.
By $K'+1-k\le  K'$, $j_k=i_{K'+1-k}\in\{i_1,\dots ,i_{K'}\}$.
By the definition of $K'$, $\{i_1,\dots ,i_{K'}\}\cap N(m')=\emptyset $.
Thus, $j_k\notin N(m')$.

Next, let $k\in\{K'+1,\dots, K'+K-K''\}$.
By $K''+k-K'>K''$, $j_k=i_{K''+k-K'}\in \{i_{K''+1},\dots ,i_K\}$.
By Claim 2, $j_k\notin N(m')$.
Hence, $\{j_1,\dots ,j_{K'+K-K''}\}\subseteq N\setminus N(m')$. \\

\noindent {\sl Proof of (a)}:\enspace  By \eqref{Step 8 i_{K'}=i}, $b_1=a_{i_{K'}}=a_i$.
Thus, by $a_i\in M_+(m')$ and Claim 4, $b_1\in M_+(m')=M_+(m')\setminus M_{++}(m')$. \\

\noindent {\sl Proof of (b)}:\enspace 
Let $k\in \{1,\dots ,K'+K-K''-1\}$.
There are three cases.
\\
\\
{\sc Case 1:} $k\le  K'-1$.

We have $j_k=i_{K'+1-k}$ and $\{b_k,b_{k+1}\}=\{a_{i_{K'+1-k}},a_{i_{K'-k}}\}$.
Thus, by \eqref{Step 8 dc},
\begin{equation*}
\{b_k,b_{k+1}\}=\{a_{i_{K'+1-k}},a_{i_{K'-k}}\}\subseteq D(R_{i_{K'+1-k}},p(m'))=D(R_{j_k},p(m')).
\end{equation*}
{\sc Case 2:} $k=K'$.

By \eqref{Step 8 i_1}, $j_{K'}=i_1=m'+1$.
Thus, $b_{K'}=a_{i_1}=a_{m'+1}$.
By $a_{m'+1}\in D(R_{m'+1},p(m'))$, $b_{K'}\in D(R_{j_{K'}},p(m'))$.
By $a_{m'+1}(R^{m'})=x_{m'+1}$, $t_{m'+1}(R^{m'})\ge  p_{m'+1}(m')$, and Fact~\ref{domination},
$$
(x_{m'+1},p_{m'+1}(m'))\,R_{m'+1}\,f_{m'+1}(R^{m'})\,R_{m'+1}\,(a_{m'+1},p_{a_{m'+1}}(m')).
$$
Thus, by $a_{m'+1}\in D(R_{m'+1},p(m'))$, $x_{m'+1}\in D(R_{m'+1},p(m'))$.
By $b_{K'+1}=a_{i_{K''}}=a_{m'}=x_{m'+1}$, $b_{K'+1}\in D(R_{j_{K'}},p(m'))$.
\\
\\
{\sc Case 3:} $k\ge  K'+1$.

We have $j_k=i_{K''+k-K'}$ and $\{b_k,b_{k+1}\}=\{a_{i_{K''+k-K'-1}},a_{i_{K''+k-K'}}\}$.
Thus, by \eqref{Step 8 dc},
\begin{equation*}
\{b_k,b_{k+1}\}=\{a_{i_{K''+k-K'-1}},a_{i_{K''+k-K'}}\}\subseteq D(R_{i_{K''+k-K'}},p(m'))=D(R_{j_k},p(m')).
\end{equation*}
{\sl Proof of (c)}:\enspace 
Note that $j_{K'+K-K''}=i_K$ and $b_{K'+K-K''}=a_{i_{K-1}}$.
Thus, by \eqref{Step 8 dc},
\begin{equation*}
b_{K'+K-K''}=a_{i_{K-1}}\in D(R_{i_K},p(m'))=D(R_{j_{K'+K-K''}},p(m')).
\end{equation*}\end{proof*}

By Claim 5 and Step~\ref{property of IRIC sequence}, $0\notin D(R_{j_{K'+K-K"}},p(m'))$.
However, by \eqref{Step 8 i_K}, $a_{j_{K'+K-K"}}=a_{i_K}=0$.
This contradicts $((a_j,t_j))_{j\in N}\in Z^{\min}(R^{m'})$, which completes the proof of step~\ref{p^{m'+1}>t_{m'+1}}.
\hfill $\blacksquare{}$ \vspace{12pt} 

For the next step, it is convenient to introduce the following notations: Let $M(0)=M_+(0)=\emptyset$.
\setcounter{equation}{0}
\begin{step}
\label{i^*}
Let $m'\in \{0,\dots ,m-1\}$.
Suppose that $M_+(m')\neq \emptyset $ or $p_{m'+1}(m')>t_{m'+1}(R^{m'})$.
Then, there is $i^*\in N\setminus N(m'+1)$ such that the following two conditions hold:\\
(a) If $p_{m'+1}(m')>t_{m'+1}(R^{m'})$, then
\begin{equation*}
D(R_{i^*},p(m'))\cap (M_+(m')\cup \{x_{m'+1}\})\neq \emptyset.
\end{equation*}%
Else, 
\begin{equation*}
D(R_{i^*},p(m'))\cap M_+(m')\neq \emptyset.
\end{equation*}
(b) If $m'\le  m-2$, then there is a set $
\{j_{m'+2},\dots ,j_m\}\subseteq N\setminus (N(m'+1)\cup \{i^*\})$ of $m-(m'+1)$ distinct
agents such that for each $k\in \{m'+2,\dots ,m\}$, $
x_k\in D(R_{j_k},p(m'))$.
\end{step}
\noindent {\bf \sl Proof of Step~\ref{i^*}\/}:\enspace 
Let $((a_i,t_i))_{i\in N}\in Z^{\min}(R^{m'})$.
For agent $m'+1$, by Fact~\ref{path}, there is a sequence $\{i_k\}_{k=1}^K$ of $K$ distinct agents such that\footnote{If $a_{m'+1}=0$, then $\{i_k\}_{k=1}^K=\{m'+1\}$ and thus, the latter part of \eqref{Step 9 i_K} and \eqref{Step 9 dc} vacuously hold.}
\begin{align}
& i_1=m'+1,  \label{Step 9 i_1}\\
&a_{i_K}=0\text{ and } \text{for each }k\in \{1,\dots, K-1\},\ a_{i_k}\ne 0\text{ and}   \label{Step 9 i_K}\\
& \text{for each }k\in \{2,\dots ,K\},\{a_{i_{k-1}},a_{i_k}\}\subseteq
D(R_{i_k}^{m'},p(m')).  \label{Step 9 dc}
\end{align}%
\textsc{Claim 1:} \emph{Let }$x_j\in M(m')\cap \{a_{i_1},\dots
,a_{i_K}\}$\emph{. Then, }$x_j\in M_+(m')$\emph{. }\\
\\
\begin{proof*}
Suppose for contradiction that $x_j\notin M_+(m')$.
Then, by Step~\ref{N(m')'s demand}~(ii) (Unique demand property), $D(R_j',p(m'))=\{x_j\}$.
Thus, $a_j=x_j$.
Therefore, by $x_j\in \{a_{i_1},\dots ,a_{i_K}\}$, $j\in \{i_1,\dots, i_K\}$.

If $j=i_1$, then by \eqref{Step 9 i_1}, $j=m'+1$, which contradicts $j\in \{1,\dots, m'\}$.
Thus $j=i_k$ for some $k\in \{2,\dots, K\}$.
However, by \eqref{Step 9 dc},
$$\{a_{i_{k-1}},a_{i_k}\}\subseteq D(R^{m'}_{i_k},p(m'))=D(R_j',p(m'))=\{x_j\},$$
a contradiction.
\end{proof*}
%
%
\\

Note that by $|N(m')|<|M(m'+1)|$ and {\sl no wastage}, there is $
i\in N\setminus N(m')$ such that $a_i\in M(m'+1)$. By
Step~\ref{N(m')'s demand}~(i) (Twin demand property),
\begin{equation}
\{a_1,\dots ,a_{m'},a_i\}=M(m'+1).  \label{Step 9 N(m') and i}
\end{equation}

Note that by Step~\ref{N(m')'s demand}~(ii) (Unique demand property), for each $x_j\in M(m')\setminus
M_+(m')$,  $a_j=x_j$.
Thus,
\begin{equation}
a_i\in M_+(m')\cup \{x_{m'+1}\}.  \label{Step 9 i's object}
\end{equation}

When $m'\le  m-2$, we define $\{j_{m'+2},\dots ,j_m\}\subseteq N$ as follows:
First note that by $((a_i,t_i))_{i\in N}\in Z^{\min}(R^{m'})$, for each $k\in \{m'+2,\dots
,m\}$, there is $i(k)\in N$ such
that $a_{i(k)}=x_k$. For each $k\in \{m'+2,\dots ,m\}$, let%
\begin{align*}
j_{k}\equiv
\begin{cases}
i_{K^{\prime }+1}&\text{if }i(k)=i_{K^{\prime }}\text{ for some }%
K^{\prime }\in \{1,\dots ,K-1\}, \\
i(k)&\text{otherwise.}
\end{cases}
\end{align*}
Note that by $a_i\in M(m'+1)$
\begin{align}
\label{Step 9 inej}
i\notin \{i(m'+2),\dots ,i(m)\}.
\end{align}

In the following claims, we show that the agents $j_{m'+2},\dots, j_m$ are distinct, $\{j_{m'+2},\dots, j_m\}\subseteq N\setminus N(m'+1)$, and for each $k\in \{m'+2, \dots, m\}$, $x_k\in D(R_{j_k}, p(m'))$.
\\
\\
\textsc{Claim 2.}
\textit{Suppose $m'\le  m-2$.
Let $k,\ell \in \{m'+2,\dots, m\}$ be such that $k\neq \ell$.
Then, $j_k\neq j_{\ell}$.}
\\
\\
\textsl{Proof: }
First suppose $i(k), i(\ell)\notin \{i_1,\dots, i_{K-1}\}$.
Then, $j_k=i(k)$ and $j_{\ell}=i(\ell)$.
By $i(k)\neq i(\ell)$, $j_k\neq j_{\ell}$.

Next, suppose $i(k)\notin  \{i_1,\dots, i_{K-1}\}$ and $ i(\ell)\in \{i_1,\dots, i_{K-1}\}$.
Then, $j_k=i(k)$, and there is $K'\in \{1,\dots, K-1\}$ such that $i(\ell)=i_{K'}$ and $j_{\ell}=i_{K'+1}$.
By $a_{i_K}=0$ and $a_{i(k)}=x_{i(k)}$, $j_k=i(k)\ne i_K$, which implies $j_k\notin \{i_1,\dots, i_K\}$.
By $j_{\ell}=i_{K'+1}$, $j_{\ell}\in \{i_1,\dots, i_K\}$, 
Thus, $j_k\neq j_{\ell}$.
The same argument holds for the converse case.

Finally, suppose $i(k),i(\ell)\in  \{i_1,\dots, i_{K-1}\}$.
There are $K',K''\in \{1,\dots, K-1\}$ such that $i_{K'}=i(k)$ and $j_k=i_{K'+1}$, and $i_{K''}=i(\ell)$ and $j_{\ell}=i_{K''+1}$.
By $i_{K'}=i(k)\neq i(\ell)=i_{K''}$, $K'\neq K''$.
This implies $K'+1\neq K''+1$.
Thus, $j_k=i_{K'+1}\neq i_{K''+1}=j_{\ell}$.
\qed
\\
\\
\noindent
\textsc{Claim 3.}
\textit{Suppose $m'\le  m-2$.
Let $k\in \{m'+2, \dots, m\}$. Then, $j_k\notin N(m'+1)$.}
\newline
\newline
\textsl{Proof: }
First, suppose $i(k)\notin \{i_1,\dots, i_{K-1}\}$.
Then $j_k=i(k)$.
By Step~\ref{N(m')'s demand}~(i) (Twin demand property) and $a_{i(k)}=x_k\notin M(m'+1)$, $j_k\notin N(m')$.
In addition, by \eqref{Step 9 i_1}, $i(k)\neq i_1= m'+1$.

Next, suppose $i(k)\in \{i_1,\dots, i_{K-1}\}$.
There is $K'\in \{1,\dots, K-1\}$ such that $i(k)=i_{K'}$ and $j_k=i_{K'+1}$.
By $i(k)=i_{K'}$, $a_{i_{K'}}=x_k$.
By \eqref{Step 9 dc}, $x_k=a_{i_{K'}}\in D(R^{m'}_{i_{K'+1}},p (m'))$.
Thus, by $x_k\notin M(m'+1)$ and Step~\ref{N(m')'s demand}~(i) (Twin demand property), $j_k=i_{K'+1}\notin N(m')$.
In addition, by $K'+1>1$ and \eqref{Step 9 i_1}, $j_k\neq i_1=m'+1$.
\qed
\\
\\
\textsc{Claim 4.}
\textit{Suppose $m'\le  m-2$.
Let $k\in \{m'+2, \dots, m\}$.
Then, $x_k\in D(R_{j_k},p(m'))$.}
\newline
\newline
\textsl{Proof: }
First, suppose $i(k)\notin \{i_1,\dots, i_{K-1}\}$.
Then $j_k=i(k)$.
By $a_{i(k)}=x_k$, $x_k\in D(R_{j_k},p(m'))$.

Next, suppose $i(k)\in \{i_1,\dots, i_{K-1}\}$.
There is $K'\in \{1,\dots, K-1\}$ such that $i(k)=i_{K'}$ and $j_k=i_{K'+1}$.
By $i(k)=i_{K'}$, $a_{i_{K'}}=x_k$.
By \eqref{Step 9 dc}, $x_k=a_{i_{K'}}\in D(R^{m'}_{i_{K'+1}},p (m'))=D(R_{j_k},p(m'))$.
\qed \\

We now complete the proof of this step by considering three disjoint cases. \\ \\
\textsc{Case 1:} \textit{$M(m'+1)\cap \{a_{i_1},\dots, a_{i_K}\}= \emptyset$.}

Let
$$i^*\equiv i.$$
By \eqref{Step 9 i's object} and $M(m'+1)\cap \{a_{i_1},\dots, a_{i_K}\}= \emptyset$,
\begin{align}
\label{Step 9 i notin}
i\notin \{i_1,\dots, i_K\}.
\end{align}
Thus, by \eqref{Step 9 i_1}, $i\neq i_1=m'+1$.
Thus, by $i\notin N(m')$, $i^*\in N\setminus N(m'+1)$.
\\
{\sl Proof of (a).}
By \eqref{Step 9 i's object} and $i^*=i$, $D(R_{i^*},p(m'))\cap (M_+(m')\cup \{x_{m'+1}\})\neq \emptyset$.
Thus, we are done when $p_{m'+1}(m')>t_{m'+1}(R^{m'})$.
Now, suppose $p_{m'+1}(m')\le  t_{m'+1}(R^{m'})$.
Since either $M_+(m')\neq \emptyset$ or $p_{m'+1}(m')>t_{m'+1}(R^{m'})$, we have $M_+(m')\neq \emptyset$.
By Steps \ref{x_ in M_+(m')} (Outside receiver I) and \ref{x_i in M_+(k)} (Outside receiver II), there is $k\in N\setminus N(m')$ such that $a_k\in M_+(m')$.
By \eqref{Step 9 N(m') and i}, $k=i$.
Thus, $a_i\in M_+(m')$.
Thus, by $i^*=i$, $D(R_{i^*},p(m'))\cap M_+(m')\neq \emptyset$.
\\
{\sl Proof of (b).}
Suppose $m'\le  m-2$.
Let $k\in\{m'+2,\dots ,m\}$.
By Claims 2, 3, and 4, we only need to show $j_k\neq i^*$.
First, suppose $i(k)\notin\{i_1,\dots ,i_K\}$.
Then, $j_k=i(k)$.
Thus, by \eqref{Step 9 inej}, $j_k=i(k)\neq i=i^*$.

Next, suppose $i(k)\in\{i_1,\dots ,i_K\}$.
There is $K'\in \{1,\dots, K-1\}$ such that $i(k)=i_{K'}$ and $j_k=i_{K'+1}$.
By $j_k\in \{i_1,\dots ,i_K\}$ and \eqref{Step 9 i notin}, $j_k\neq i=i^*$.
\\
\\
\textsc{Case 2:}
\textit{$M(m')\cap \{a_{i_1},\dots, a_{i_K}\}\neq \emptyset$ and $x_{m'+1}\notin \{a_{i_1},\dots, a_{i_K}\}$.}

By $M(m')\cap \{a_{i_1},\dots, a_{i_K}\}\neq \emptyset$ and \eqref{Step 9 i_K}, there is $K'\in \{1,\dots ,K-1\}$ such that
\begin{align}
& a_{i_{K'}}\in M(m')\text{ and }  \label{in M}\\
& \{a_{i_{K'+1}},\dots ,a_{i_K}\}\cap M(m')=\emptyset .
\end{align}

Let
$$i^*\equiv i_{K'+1}.$$
By $a_{i_{K'+1}}\notin M(m')$ and Step~\ref{N(m')'s demand}~(i) (Twin demand property), $i_{K'+1}\notin N(m'-1)$.
If $i_{K'+1}=m'$, then by $a_{i_{K'+1}}\notin M(m')$ and Step~\ref{N(m')'s demand}~(i) (Twin demand property), $a_{i_{K'+1}}=x_{m'+1}$, which is impossible since $x_{m'+1}\notin \{a_{i_1},\dots, a_{i_K}\}$.
Thus, $i_{K'+1}\neq m'$.
Finally, by $K'+1>1$ and \eqref{Step 9 i_1}, $i_{K'+1}\neq i_1=m'+1$.
Hence, $i^*=i_{K'+1}\in N\setminus N(m'+1)$.
\\
{\sl Proof of (a).}
By \eqref{Step 9 dc} and $i^*=i_{K'+1}$, $a_{i_{K'}}\in D(R^{m'}_{i_{K'+1}},p(m'))=D(R_{i^*},p(m'))$.
By Claim 1 and $a_{i_{K'}}\in M(m')$, $a_{i_{K'}}\in M_+(m')$.
Thus, $D(R_{i^*},p(m'))\cap M_+(m')\neq \emptyset$.
\\
{\sl Proof of (b).}
Suppose $m'\le  m-2$.
Let $k\in\{m'+2,\dots ,m\}$.
By Claims 2, 3, and 4, we only need to show $j_k\neq i^*$.
First, suppose $i(k)\notin\{i_1,\dots ,i_K\}$.
Then, $j_k=i(k)$.
By $i(k)\notin\{i_1,\dots ,i_K\}$ and $i^*=i_{K'+1}\in \{i_1,\dots ,i_K\}$, $j_k\neq i^*$.

Next, suppose $i(k)\in\{i_1,\dots ,i_K\}$.
There is $K^*\in \{1,\dots ,K-1\}$ such that $i(k)=i_{K^*}$ and $j_k=i_{K^*+1}$.
By $a_{i_{K'}}\in M(m')$ and $a_{i_{K^*}}=a_{i(k)}=x_k\notin M(m'+1)$, $K'\neq K^*$.
Thus, $j_k=i_{K^*+1}\neq i_{K'+1}=i^*$.
\\
\\
\textsc{Case 3:}
\textit{$x_{m'+1}\in \{a_{i_1},\dots, a_{i_K}\}$.}

Let $K'\in \{1,\dots, K\}$ be such that
$$a_{i_{K'}}=x_{m'+1}.$$
By \eqref{Step 9 i_K}, $K'<K$.
Thus, $i_{K'+1}$ exists.
Let
$$i^*\equiv i_{K'+1}.$$
We show $i^*\in N\setminus N(m'+1)$.
By $a_{i_{K'}}=x_{m'+1}$, \eqref{Step 9 dc}, and Step~\ref{N(m')'s demand}~(i) (Twin demand property), $i_{K'+1}\notin N(m'-1)$.
By $K'+1>1$ and \eqref{Step 9 i_1}, $i_{K'+1}\neq i_1=m'+1$.

Note that by \eqref{Step 9 N(m') and i}, \eqref{Step 9 i's object}, and Step~\ref{N(m')'s demand}~(i) (Twin demand property),
$$a_i=x_{m'+1}\text{ or }a_{m'}=x_{m'+1}.$$

Suppose $a_{m'}=x_{m'+1}$.
Then $i_{K'}=m'$ and thus $i_{K'+1}\neq m'$.
Suppose $a_i=x_{m'+1}$.
This implies $i\ne i^*$.
If $M_+(m')\neq \emptyset$, then Steps \ref{x_ in M_+(m')} (Outside receiver I) and \ref{x_i in M_+(k)} (Outside receiver II) imply that $a_j\in M_+(m')$ for some $j\in N\setminus N(m')$, and by \eqref{Step 9 N(m') and i}, $j=i$, contradicting $a_i=x_{m'+1}$.
Thus, 
$M_+(m')=\emptyset$.
Then, by Claim 1, $a_{i_{K'+1}}\notin M(m')$.
Moreover, by $a_{i_{K'}}=x_{m'+1}$, $a_{i_{K'+1}}\neq x_{m'+1}$.
Thus, by Step~\ref{N(m')'s demand}~(i) (Twin demand property), $i_{K'+1}\neq m'$.
Hence, $i^*=i_{K'+1}\in N\setminus N(m'+1)$.
\\
{\sl Proof of (a).}
First, we show that $p_{m'+1}(m')>t_{m'+1}(R^{m'})$.
Suppose $M_+(m')\neq \emptyset$.
Then, as we have seen in the above paragraph, we can show that $a_i\ne x_{m'+1}$.
Since either $a_i=x_{m'+1}$ or $a_{m'}=x_{m'+1}$, $a_{m'}=x_{m'+1}$.
By $x_{m'+1}\in \{a_{i_1},\dots, a_{i_K}\}$, $m'\in \{i_1,\dots, i_K\}$.
Then, by Step~\ref{p^{m'+1}>t_{m'+1}}, $p_{m'+1}(m')>t_{m'+1}(R^{m'})$.
Since $M_+(m')\neq \emptyset$ or $p_{m'+1}(m')>t_{m'+1}(R^{m'})$, we can conclude that $p_{m'+1}(m')>t_{m'+1}(R^{m'})$.

By $a_{i_{K'}}=x_{m'+1}$ and \eqref{Step 9 dc}, $x_{m'+1}\in D(R^{m'}_{i_{K'+1}},p(m'))=D(R_{i^*}, p(m'))$.
Thus, $D(R_{i^*},p(m'))\cap (M_+(m')\cup \{x_{m'+1}\})\neq \emptyset$.
\\
{\sl Proof of (b).}
Suppose $m'\le  m-2$.
Let $k\in\{m'+2,\dots ,m\}$.
By Claims 2, 3, and 4, we only need to show $j_k\neq i^*$.
First, suppose $i(k)\notin\{i_1,\dots ,i_K\}$.
Then, $j_k=i(k)$.
By $i(k)\notin\{i_1,\dots ,i_K\}$ and $i^*=i_{K'+1}\in \{i_1,\dots ,i_K\}$, $j_k\neq i^*$.

Next, suppose $i(k)\in\{i_1,\dots ,i_K\}$.
There is $K^*\in \{1,\dots ,K-1\}$ such that $i(k)=i_{K^*}$ and $j_k=i_{K^*+1}$.
By $a_{i_{K^*}}=a_{i(k)}=x_k\neq x_{m'+1}=a_{i_{K'}}$, $K^*\neq K'$.
Thus, $j_k=i_{K^*+1}\neq i_{K'+1}=i^*$.
\hfill $\blacksquare{}$ \vspace{12pt} 
\setcounter{equation}{0}
\begin{step}
\label{f_j P_i 0}
There are $i\in N\setminus N(m)$ and $j\in
\{1,\dots ,m\}$ such that $f_j(R^{j-1})\,P_i\,(0,0)$.
\end{step}
\noindent {\bf \sl Proof of Step~\ref{f_j P_i 0}\/}:\enspace 
The proof consists of two substeps. \\
\\
\textsc{Substep 10-1} \textit{Let $m'\in \{0,1,\dots ,m-2\}$ be
such that $M_+(m')\neq \emptyset $ or $p_{m'+1}(m')>t_{m'+1}(R^{m'})$. Then, $
M_+(m'+1)\neq \emptyset $.
} \\
\\
\begin{proof}
Suppose for contradiction that $M_+(m'+1)=\emptyset.$
By $M_+(m')\neq \emptyset $ or $p_{m'+1}(m')>t_{m'+1}(R^{m'})$, Step~\ref{i^*}~(a) implies that there is $i^*\in N\setminus N(m'+1)$ such that if $p_{m'+1}(m')>t_{m'+1}(R^{m'})$, then
\begin{equation}
\label{Step 10-1 x_{m'+1}}
D(R_{i^*},p(m'))\cap (M_+(m')\cup \{x_{m'+1}\})\neq \emptyset ,
\end{equation}
else,
\begin{equation}
\label{Step 10-1 M_+(m')}
D(R_{i^*},p(m'))\cap M_+(m')\neq \emptyset .
\end{equation}
Moreover, by Step~\ref{i^*}~(b) and $m'\le m-2$, there is a set $\{j_{m'+2},\dots ,j_m\}\subseteq N\setminus (N(m'+1)\cup\{i^*\})$ of $m-(m'+1)$
distinct agents such that for each $k\in \{m'+2,\dots ,m\}$, $x_k\in D(R_{j_k},p(m'))$.

Let $((a_i,t_i))_{i\in N}\in Z^{\min}(R^{m'+1})$.
Note that by $M_+(m'+1)=\emptyset$ and Step~\ref{N(m')'s demand}~(ii) (Unique demand property),
\begin{equation}
\label{Step 10-1 x_i=x_i}
a_i=x_i\text{ for each }i\in N(m'+1).
\end{equation}
\textsc{Claim 1:} \textit{Let $i\in N\setminus N(m'+1)$.
Suppose that there is $x_k\in M$ such that $x_k\in D(R_i,p(m'))$ and $p_k(m'+1)<p_k(m')$.
Then, $p_{a_i}(m'+1)<p_{a_i}(m')$.} \\
\\
\begin{figure}[!hbt]
\centering
\includegraphics[width=0.7\linewidth]{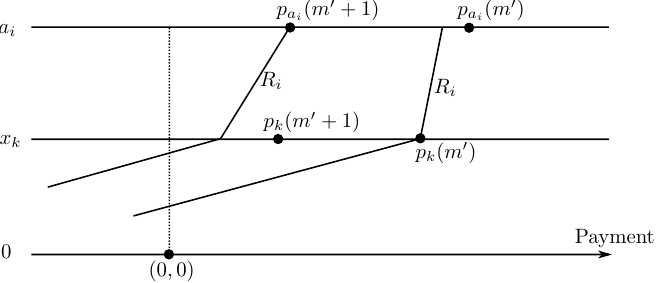}
\caption{Illustration of the proof of Claim~1.}
\label{fig:Step10_1}
\end{figure}
\\
\begin{proof*}
(See Figure~\ref{fig:Step10_1} for illustration.)
Note that
\begin{equation*}
(a_i,p_{a_i}(m'+1))~R_i~(x_k,p_k(m'+1))~P_i~(x_k,p_k(m'))~R_i~(a_i,p_{a_i}(m')),
\end{equation*}
where the first relation follows from $a_i\in D(R_i,p(m'+1))$, the second from $p_k(m'+1)<p_k(m')$, and the last from $x_k\in D(R_i,p(m'))$.
This implies $p_{a_i}(m'+1)<p_{a_i}(m')$. \end{proof*}

\noindent
\textsc{Claim 2:} \textit{For each $K\in \mathbb{N}$, there exists a set $
N_K\equiv \{i_1,\dots ,i_K\}\subseteq N\setminus N(m'+1)$ of $
K$ distinct agents such that $i^*\in N_K$ and for each $a\in \{a_{i_1},\dots ,a_{i_K}\}$, $p_a(m'+1)<p_a(m')$.
}\\
\\
\begin{proof*} The proof is by induction. \\
\textit{Induction base.}
Let $K=1$ and $N_1=\{i^*\}$. By $i^*\in N\setminus N(m'+1)$, $N_1\subseteq N\setminus N(m'+1)$.
Note that if we show that there is $a\in D(R_{i^*},p(m'))$ such that $p_a(m'+1)<p_a(m')$, then Claim~1 implies that $p_{a_{i^*}}(m'+1)<p_{a_{i^*}}(m')$.

First, suppose that there is $x_k\in M_+(m')$ such that $x_k\in D(R_{i^*},p(m'))$.
By $M_+(m'+1)=\emptyset$, $p_k(m'+1)\le  t_k(R^{k-1})$.
By $x_k\in M_+(m')$, $p_k(m')>t_k(R^{k-1})$.
Thus, $p_k(m'+1)<p_k(m')$.

Next, suppose that $D(R_{i^*},p(m'))\cap M_+(m')=\emptyset $.
By \eqref{Step 10-1 M_+(m')} we have  $p_{m'+1}(m')>t_{m'+1}(R^{m'})$.
Then, by \eqref{Step 10-1 x_{m'+1}}, $x_{m'+1}\in D(R_{i^*},p(m'))$.
By $M_+(m'+1)=\emptyset$, $p_{m'+1}(m'+1)\le  t_{m'+1}(R^{m'})$.
Thus, $p_{m'+1}(m'+1)<p_{m'+1}(m')$.
\\
\textit{Induction argument.}
Let $K\ge 1$ and assume that there is a set $N_K=\{i_1,\dots ,i_K\}\subseteq N\setminus N(m'+1)$ of $K$
distinct agents such that
\begin{align}
\label{Step 10-1 i^* in N_K}
& i^*\in N_K,  \text{ and}  \\
\label{Step 10-1 p_a(m'+1)<p_a(m')}
& \text{for each }a\in \{a_{i_1},\dots ,a_{i_K}\},\ p_a(m'+1)<p_a(m').
\end{align}

By \eqref{Step 10-1 x_i=x_i} and $N_K\subseteq N\setminus N(m'+1)$, $\{a_{i_1},\dots ,a_{i_K}\}\subseteq M\setminus M(m'+1) $.
Without loss of generality, assume that $\{a_{i_1},\dots,a_{i_K}\}=\{x_{m'+2},\dots, x_{m'+K+1}\}$.
By $i^*\in N_K$ and $i^*\notin \{j_{m'+2},\dots ,j_{m'+K+1}\}$, there is $j_k\in \{j_{m'+2},\dots ,j_{m'+K+1}\}$ such that $j_k\notin N_K$.

Let $N_{K+1}\equiv N_K\cup \{j_k\}$.
By $N_K\subseteq N\setminus N(m'+1)$ and $j_k\notin N(m'+1)$, $N_{K+1}\subseteq N\setminus N(m'+1)$.
By \eqref{Step 10-1 i^* in N_K}, $i^*\in N_K\subseteq N_{K+1}$.
Note that by $j_k\in \{j_{m'+2},\dots ,j_{m'+K+1}\}$, $x_k\in \{x_{m'+2},\dots x_{m'+K+1}\}=\{a_{i_1},\dots ,a_{i_K}\}$.
Thus, by \eqref{Step 10-1 p_a(m'+1)<p_a(m')}, $p_k(m'+1)<p_k(m')$.
Thus, by $x_k\in D(R_{j_k},p(m'))$, Claim 1 implies that $p_{a_{j_k}}(m'+1)<p_{a_{j_k}}(m')$.
This completes the proof of Claim~2.
\end{proof*}

Claim 2 completes the proof of Substep 10-1 because $N$ is finite. 
\end{proof}
\\
\textsc{Substep 10-2 } \textit{Completing the proof of Step 10.} \\

Remember that $p_1(0)>t_1(R^{0})$. Thus, by using Substep 10-1 repeatedly, we have $M_+(m-1)\neq \emptyset $.
Then, by Step~\ref{i^*}~(a), there is $i\in N\setminus N(m) $ such that if $p_m(m-1)>t_m(R^{m-1}) $, then
\begin{equation}
\label{Step 10-2 x_m}
D(R_i,p(m-1))\cap (M_+(m-1)\cup \{x_m\})\neq \emptyset ,
\end{equation}%
else,
\begin{equation}
\label{Step 10-2 M_+(m-1)}
D(R_i,p(m-1))\cap M_+(m-1)\neq \emptyset .
\end{equation}

First, suppose that there is $x_j\in M_+(m-1)$ such that $x_j\in D(R_i,p(m-1))$.
By $x_j\in M_+(m-1)$, $p_j(m-1)>t_j(R^{j-1})$.
Thus, by $x_j\in D(R_i,p(m-1))$, $f_j(R^{j-1})\,P_i\,(x_j,p_j(m-1))~R_i~(0,0)$.

Next, suppose that $D(R_i,p(m-1))\cap M_+(m-1)=\emptyset $.
By \eqref{Step 10-2 x_m} and \eqref{Step 10-2 M_+(m-1)}, $p_m(m-1)>t_m(R^{m-1})$ and $x_m\in D(R_i,p(m-1))$.
Thus, $f_m(R^{m-1})\,P_i\,(x_m,p_m(m-1))~R_i~(0,0)$.
\hfill $\blacksquare{}$ \vspace{12pt} 
\setcounter{equation}{0}
\begin{step}Completing the proof.
\end{step}

First, we show the following. \\
\\
\textsc{Claim:} \textit{Let $i\in N(m)$. Then, $a_i(R^{m})\in\{x_i,x_{i+1}\}$.} \\
\\
\begin{proof*} Suppose for contradiction that $a_i(R^{m})\notin\{x_i,x_{i+1}\}$. Then,
\begin{align}
V^{R'_i}(x_{i+1},f_i(R^{m}))  &\ge   V^{R'_i}(x_{i+1},(a_i(R^{m}),0))
\tag*{by \textsl{no subsidy}} \\
&> \overline{V}  \tag*{by ($i$-iii) in Step~\ref{pref}} \\
&> p_{i+1}(m).  \tag*{by Step~\ref{less than V}}
\end{align}%
This implies $(x_{i+1},p_{i+1}(m))\,P_i'\,f_i(R^{m})$, contradicting Fact~\ref{domination}. \end{proof*} \\

Since $R_m'$ is $f_m(R^{m-1})$-favoring,
Lemma~\ref{sp and fav} implies $f_m(R^{m})=f_m(R^{m-1})$.
Thus, by $a_m(R^{m-1})=x_m$, $a_m(R^{m})=x_m$. Therefore, by applying
Claim repeatedly, we obtain
\begin{equation*}
a_i(R^{m})=x_i\text{ for each }i\in N(m).
\end{equation*}%
Thus, by \textsl{individual rationality},
\begin{equation}
\label{Step 11 t_i leq V'_i}
t_i(R^{m})\le  V^{R'_i}(x_i,(0,0))\text{ for each }%
i\in N(m).
\end{equation}

By Step~\ref{f_j P_i 0}, there are $i\in N\setminus N(m)$ and $j\in \{1,\dots ,m\}$ such that $f_j(R^{j-1})\,P_i\,(0,0)$.
Note that this implies that $(\{i\}, \{x_j\})$ is an IRIC from $f_j(R^{j-1})$. 
By \eqref{Step 11 t_i leq V'_i} and ($j$-ii) in Step~\ref{pref}, $t_j(R^{m})\le  V^{R'_j}(x_j,(0,0))<d(f_j(R^{j-1}))$.
Therefore, $(\{i\}, \{x_j\})$ is an IRIC from $f_j(R^{m})$, and thus, $f_j(R^{m})\,P_i\,(0,0)$

Let $R_j''\equiv R_i$. By \textsl{strategy-proofness} and
$f_j(R^{m})\,P_i\,(0,0)$, $f_j(R''_j,R_{-j}^{m})~R_j''~f_j(R^{m})~P''_j~(0,0)$. By \textsl{equal treatment of equals} and $
R_i=R_j''$, $f_i(R_j'',R_{-j}^{m})\,I_i\,f_j(R_j'',R_{-j}^{m})\,P_i\,%
(0,0)$. Thus, by Lemma~\ref{zero payment},
\begin{equation*}
a_j(R_j'',R_{-j}^{m})\neq 0\text{ and }%
a_i(R_j'',R_{-j}^{m})\neq 0.
\end{equation*}

By $a_i(R_j'',R_{-j}^{m})\neq 0$ and $|N(m)|=m$, there is $k\in N(m)$ such that $a_k(R_j'',R_{-j}^{m})=0$.
By $a_j(R_j'',R_{-j}^{m})\neq 0$, $k\neq j$. By Lemma~\ref{zero payment}, $t_k(R_j'',R_{-j}^{m})=0$.
Let $p\equiv p^{\min }(R_j'',R_{-j}^{m})$.
Similarly to Step~\ref{less than V}, we can show $p_{x_{k+1}}<\bar{V}$.
By (k-$ii$), $p_{x_{k+1}}<\bar{V}<V^{R'_k}(x_{k+1},(0,0))$.
Thus, we have $(x_{k+1},p_{x_{k+1}})~P_k'~(0,0)=f_k(R''_j,R_{-j}^{m})$.
This contradicts Fact~\ref{domination} \hfill $\blacksquare $

\bibliographystyle{ecta}
\bibliography{order}

\end{document}